\documentclass[12pt]{amsart}
\usepackage{latexsym,amssymb,amsfonts,amsmath,mathrsfs,float, hyperref}
\usepackage{thmtools,thm-restate}
\usepackage{xcolor}
\usepackage{tikz}
\hypersetup{
    colorlinks,
    linkcolor={black},
    citecolor={black},
}

\calclayout

\newcommand{\parset}{
	\setlength{\parskip}{3mm}
  	\setlength{\parindent}{0mm}}
\parset
	
  \renewcommand{\Pr}{\mbox{\rm Pr}}
  \newcommand{\Exp}{{\mathbb{E}}}
  \newcommand{\E}{\mathbb{E}}

  \DeclareMathOperator{\Stab}{Stab}\DeclareMathOperator{\Spec}{Spec}
  
  \newcommand{\C}{\mathbb{C}} 
  \newcommand{\D}{\mathbb{D}} 
  \newcommand{\Z}{\mathbb{Z}} 
  \newcommand{\F}{\mathbb{F}} 
  \newcommand{\Lcal}{\mathcal{L}}
  \newcommand{\pmset}[1]{\{-1,1\}^{#1}} 
  \newcommand{\bset}[1]{\{0,1\}^{#1}} 
  \DeclareMathOperator{\vspan}{Span} 

  \newcommand{\one}{\mathbf{1}}

  \DeclareMathOperator{\supp}{supp}

  \newcommand{\eps}{\varepsilon}

  \newcommand{\poly}{\mbox{\rm poly}}

  \DeclareMathOperator{\Diag}{Diag}
  \DeclareMathOperator{\spec}{Spec}
  \DeclareMathOperator{\Est}{Est}
  \DeclareMathOperator{\FourEst}{FourEst}
  \DeclareMathOperator{\dist}{dist}

  \newcommand{\A}{\mathcal A}

  \newcommand{\beq}{\begin{equation}}
  \newcommand{\eeq}{\end{equation}}
  \newcommand{\beqn}{\begin{equation*}}
  \newcommand{\eeqn}{\end{equation*}}
  \newcommand{\beqr}{\begin{eqnarray}}
  \newcommand{\eeqr}{\end{eqnarray}}
  \newcommand{\beqrn}{\begin{eqnarray*}}
  \newcommand{\eeqrn}{\end{eqnarray*}}
  \newcommand{\bmline}{\begin{multline}}
  \newcommand{\emline}{\end{multline}}
  \newcommand{\bmlinen}{\begin{multline*}}
  \newcommand{\emlinen}{\end{multline*}}
  
  \theoremstyle{plain}
  \newtheorem{theorem}{Theorem}[section]
  \newtheorem{lemma}[theorem]{Lemma}
  \newtheorem{proposition}[theorem]{Proposition}
  
  \newtheorem{corollary}[theorem]{Corollary}
  
  \theoremstyle{definition}
  \newtheorem{definition}[theorem]{Definition}
  \newtheorem{algorithm}{Algorithm}[section]

  \theoremstyle{remark}

  \renewenvironment{proof}[1][]{
    	\begin{trivlist}
     	\item[\hspace{\labelsep}{\em\noindent Proof#1:\/}]}
     	{{\hfill$\Box$}
    	\end{trivlist}
  }

\begin{document}
\title{A near-optimal Quadratic Goldreich--Levin algorithm}
\author{Jop Bri\"{e}t}
\address{CWI \& QuSoft, Amsterdam, Netherlands}
\email{j.briet@cwi.nl}

\author{Davi Castro-Silva}
\address{University of Cambridge, Cambridge, UK}
\email{dd654@cam.ac.uk}

\begin{abstract}
    In this paper, we give a quadratic Goldreich-Levin algorithm that is close to optimal in the following ways.
    Given a bounded function~$f$ on the Boolean hypercube~$\F_2^n$ and any~$\eps>0$, the algorithm returns a quadratic polynomial $q: \F_2^n \to \F_2$ so that the correlation of~$f$ with the function~$(-1)^q$ is within an additive~$\eps$ of the maximum possible correlation with a quadratic phase function.
    The algorithm runs in~$O_\eps(n^3)$ time and makes $O_\eps(n^2\log n)$ queries to~$f$, which matches the information-theoretic lower bound of~$\Omega(n^2)$ queries up to a logarithmic factor.

    As a result, we obtain a number of corollaries:
    \begin{itemize}
        \item A near-optimal self-corrector of quadratic Reed-Muller codes, which makes $O_\eps(n^2\log n)$ queries to a Boolean function~$f$ and returns a quadratic polynomial~$q$ whose relative Hamming distance to~$f$ is within~$\eps$ of the minimum distance.
        \item An algorithmic polynomial inverse theorem for the order-3 Gowers uniformity norm.
        \item An algorithm that makes a polynomial number of queries to a bounded function~$f$ and decomposes~$f$ as a sum of~$\poly(1/\eps)$ quadratic phase functions and error terms of order~$\eps$.
    \end{itemize}

    Our algorithm is obtained using ideas from recent work on quantum learning theory.
    Its construction deviates from previous approaches based on algorithmic proofs of the inverse theorem for the order-3 uniformity norm (and in particular does not rely on the recent resolution of the polynomial Fre\u{\i}man-Ruzsa conjecture).
\end{abstract}

\maketitle

\section{Introduction}

Fourier analysis plays an indispensable role in theoretical computer science.
A celebrated example of its use is in the analysis of the property-testing algorithm of Blum, Luby and Rubinfeld~\cite{blum1993self}, which uses a constant number of queries to a function $f:\F_2^n\to\F_2$ to decide with good probability whether~$f$ is close to some linear function or far from all linear functions.
In effect, this algorithm tests whether the phase function~$(-1)^{f(x)}$ has a Fourier coefficient of large magnitude.
The famous Goldreich-Levin algorithm goes a step further and allows one to \emph{learn} with high precision the entire collection of large Fourier coefficients, using~$O(n\log n)$ queries~\cite{GoldreichLevin1989}.
Many important properties of Boolean functions are effectively captured by this set of coefficients~\cite{ODonnell2014}.

A fundamental application of these results concerns the theory of error correcting codes.
The \emph{Hadamard code} consists of the $2^n$-bit strings given by the complete evaluations of the set of linear functions over~$\F_2^n$.
The BLR-test then gives a constant-query algorithm that decides whether or not a given string is close to a Hadamard codeword.
In turn, the Goldreich-Levin algorithm allows one to efficiently decode a corrupted Hadamard codeword, and in fact it represents the first instance of a list-decoding algorithm~\cite[Section~3.5]{ODonnell2014}.

Fourier analysis also plays an important role in additive combinatorics, where it is used to study additive patters in dense sets of integers or subsets of finite vector spaces~\cite{TaoVu2006, Zhao2023}.
Szemer\'edi's theorem, a centerpiece of the area, asserts that any dense set of the integers contains arithmetic progressions of arbitrary finite length~\cite{Szemeredi1975}.
While standard tools from Fourier analysis are sufficient to prove this result for 3-term arithmetic progressions~\cite{Roth1953}, they are not enough to control the count of higher-order progressions.
A new proof of Szemer\'edi's theorem due to Gowers introduced many additional ideas that resulted in a new field of higher-order Fourier analysis~\cite{Gowers1998,Gowers2001, Tao2012}.
Much of the area revolves around the so-called uniformity norms. 
Roughly speaking, the uniformity norm~$\|f\|_{U^k}$ of order~$k$ measures the degree of oscillation of a (complex-valued) function~$f$ after it has been derived~$k$ times in random directions.
Intuitively, if such derivatives are approximately constant, then the function in some sense resembles a polynomial of degree less than~$k$.
Over finite vector spaces, this gives rise to the analysis of functions in terms of polynomial phases such as~$(-1)^{P(x)}$, where~$P:\F_2^n\to\F_2$ is a low-degree polynomial.
Deep inverse theorems show that a bounded function has non-negligible order-$k$ uniformity norm if and only if it correlates with a polynomial phase of degree at most~$k-1$~\cite{BergelsonTaoZiegler2010, TaoZiegler2012}.\footnote{In fields of small characteristic, the set of polynomial phases must be generalized to \emph{non-classical polynomial phases}, which will be relevant in this work as well.}

Connections between theoretical computer science and additive combinatorics where recognized early on and facilitated a steady exchange of ideas back and forth~\cite{Trevisan2009, Lovett2017,HatamiHatamiLovett2019}.
The inverse theorems for the uniformity norms are strongly connected to property-testing algorithms for low-degree polynomials.
In the context of testing closeness to quadratic Reed-Muller codes, building on ideas from~\cite{Gowers1998}, Samorodnitsky proved a $U^3$-inverse theorem for bounded functions over~$\F_2^n$, giving a quadratic analogue of the BLR linearity test~\cite{Samorodnitsky2007}; similar results were proved by Green and Tao for fields of larger characteristic and for cyclic groups~\cite{GreenTao2008}.
Define
\begin{equation*}
    \|f\|_{u^3} = \max \big|\Exp_{x\in \F_2^n}f(x)(-1)^{q(x)}\big|,
\end{equation*}
where the maximum is taken over quadratic polynomials~$q:\F_2^n\to \F_2$.
Samorodnitsky's result may be expressed as the inequality 
\begin{equation*}
\|f\|_{u^3} \geq \exp\big(-\poly(1/\|f\|_{U^3})\big),    
\end{equation*} 
valid over all functions $f: \F_2^n \to [-1, 1]$.

By giving algorithmic versions of the ideas used in the proof of this result, Tulsiani and Wolf obtained a quadratic version of the Goldreich-Levin algorithm~\cite{TulsianiW2014}.
Given a function~$f:\F_2^n\to[-1,1]$ such that~$\|f\|_{U^3} \geq \eps$, their algorithm makes~$\exp(-\poly(1/\eps)) n^3\log n$ queries to~$f$ and, with good probability, returns a quadratic polynomial~$q:\F_2^n\to\F$ such that 
\begin{equation*}
    \big|\Exp_{x\in \F_2^n}f(x)(-1)^{q(x)}\big|
    \geq
    \exp\big(-\poly(\tfrac{1}{\eps})\big).
\end{equation*}

The quantitative aspects of this algorithm are closely related to those of the inverse theorems of the uniformity norms.
Following work of Sanders that made substantial progress towards establishing the so-called \emph{polynomial Fre\u{\i}man-Ruzsa conjecture} (PFR) ~\cite{Sanders2012}, the exponential dependence on~$\eps$ in the algorithm was improved to quasi-polynomial by Ben-Sasson et al.~\cite{Ben-Sasson2014}.
A recent breakthrough of Gowers, Green, Manners and Tao resolved the PFR conjecture~\cite{GowersGMT2025}. But while this result implies a polynomial inverse theorem for the~$U^3$-norm, since there is no known algorithmic version of it, it has not yet led to an improved quadratic Goldreich-Levin algorithm.

The main contribution of this paper (Theorem~\ref{thm:main}) immediately gives such a polynomial quadratic Goldreich-Levin algorithm.
Perhaps surprisingly, our algorithm is not obtained via an algorithmic version of PFR but instead uses a completely new approach inspired by recent work on quantum computing.
In fact, our algorithm does not rely on the assumption that the $U^3$-norm is large, but instead simply returns a near-maximal quadratic correlator.

\subsection{Our results}

One perspective on the Goldreich-Levin algorithm is that it learns an $\eps$-maximal linear phase correlator for~$f$, that is, a linear phase $x\mapsto (-1)^{a\cdot x}$ such that
$$\big|\Exp_{x\in \F_2^n} f(x) (-1)^{a\cdot x}\big| > \max_{b\in \F_2^n} \big|\Exp_{x\in \F_2^n} f(x) (-1)^{b\cdot x}\big| - \eps,$$
while making only $n\log n\, \poly(1/\eps)$ queries to~$f$ and taking $n\log n\, \poly(1/\eps)$ running time.
This algorithm can be easily transformed into its usual ``list-decoding'' variant where one learns \emph{all} linear phases that correlate at least~$\eps$ with the given function~$f$, due to Parseval's identity.

In the quadratic setting, we no longer have an analogue of Parseval's identity, and in fact there can be $\exp(n)$-many quadratic phases that have nonnegligible correlation with~$f$.
However, it might still be possible to efficiently obtain an $\eps$-maximal quadratic correlator, and this is what a \emph{quadratic Goldreich-Levin} algorithm should obtain.

Our main result is the following:

\begin{theorem}[Quadratic Goldreich-Levin] \label{thm:main}
    Let ~$f: \F_2^n \to [-1, 1]$ be a $1$-bounded function and let $\eps,\, \delta>0$.
    There is a randomized algorithm~$\A$ that makes $n^2 \log n\, \log(1/\delta) (1/\eps)^{O(\log(1/\eps))}$ queries to~$f$ and, with probability at least $1-\delta$, outputs a quadratic polynomial $p: \F_2^n \to \F_2$ such that
    $$\big|\Exp_{x\in \F_2^n} f(x) (-1)^{p(x)}\big| > \max_{q \text{ quadratic}} \big|\Exp_{x\in \F_2^n} f(x) (-1)^{q(x)}\big| - \eps.$$
    In addition to the~$\tilde{O}(n^2)$ queries, the algorithm~$\A$ runs in time $O(n^3)$.
\end{theorem}

This result is essentially optimal in two different ways.
First, it obtains an $\eps$-optimal quadratic correlator for any given $\eps>0$;
this is in contrast to previous approaches, which, when given a function~$f$ satisfying $\|f\|_{u^3} \geq \tau$, could only generate a quadratic phase that has correlation either $\exp(-\poly(1/\tau))$ \cite{TulsianiW2014} or $\exp(-\poly\log(1/\tau))$ \cite{Ben-Sasson2014}.
(Note that it would be impossible to guarantee the exact optimal correlator using only a polynomial number of queries to~$f$.)
Moreover, it is easy to see that $\Omega(n^2)$ queries to~$f$ are \emph{necessary} even to obtain a (say) $0.1$-optimal quadratic correlator with probability at least $0.1$;
our algorithm matches this trivial lower bound in query complexity up to a logarithm factor.

As an easy corollary of our main result, we obtain an efficient (and essentially optimal) self-corrector algorithm for quadratic Reed-Muller codes that is agnostic to the error rate:

\begin{corollary}[Optimal self-correction of quadratic Reed-Muller codes] \label{cor:ReedMuller}
    There is a query algorithm~$\A$ with the following guarantees.
    Given $\eps>0$ and query access to a Boolean function $f: \F_2^n\to \F_2$, $\A$ makes $n^2 \log n/\eps^{O(\log(1/\eps))}$ queries to~$f$ and, with probability at least~$2/3$, outputs a quadratic polynomial $p: \F_2^n \to \F_2$ satisfying
    $$\dist(f,\, p) < \min_{q \text{ quadratic}} \dist(f,\, q) + \eps,$$
    where $\dist$ denotes the normalized Hamming distance.
    In addition to the $\tilde{O}_{\eps}(n^2)$ queries, algorithm~$\A$ has runtime $O_{\eps}(n^3)$.
\end{corollary}

Combining our main theorem with the recent resolution of Marton's conjecture by Gowers, Green, Manners and Tao \cite{GowersGMT2025}, we easily obtain an algorithmic inverse theorem for the Gowers $U^3$-norm with polynomial bounds:

\begin{corollary}[Algorithmic polynomial Gowers inverse theorem] \label{cor:PGI}
    Let $\gamma>0$ and let $f: \F_2^n \to [-1,\, 1]$ be a function with $\|f\|_{U^3} \geq \gamma$.
    There is a randomized algorithm~$\A$ that makes $n^2 \log n/\gamma^{O(\log(1/\gamma))}$ queries to~$f$, runs in time $O(n^3)$ and that, with probability at least~$2/3$, outputs a quadratic polynomial $p: \F_2^n \to \F_2$ satisfying
    $\big|\Exp_{x\in \F_2^n} f(x) (-1)^{p(x)}\big| \geq (\gamma/2)^C$
    where $C>0$ is some absolute constant.
\end{corollary}

Finally, combining our results with the framework developed by Tulsiani and Wolf \cite{TulsianiW2014}, we obtain an algorithmic quadratic decomposition theorem which efficiently decomposes a bounded function~$f$ into a sum of $\poly(1/\eps)$-many quadratic phase function, plus errors of $U^3$-norm and $L_1$-norm at most~$\eps$.

\begin{corollary}[Efficient quadratic decomposition] \label{cor:decomposition}
    There is a randomized algorithm that, given $\eps>0$ and any $1$-bounded function $f: \F_2^n \to [-1,\, 1]$, outputs with probability at least $2/3$ a decomposition
    $$f = c_1 (-1)^{p_1(\cdot)} + \dots + c_r (-1)^{p_r(\cdot)} + g + h$$
    where the~$c_i$ are constants, the~$p_i$ are quadratic polynomials, $r \leq \poly(1/\eps)$, $\|g\|_{U^3} \leq \eps$ and $\|h\|_1 \leq \eps$.
    The algorithm makes $n^2 \log n/\eps^{O(\log(1/\eps))}$ queries to~$f$ and runs in time $O_{\eps}(n^3)$.
\end{corollary}

\subsection{QGL and stabilizer learning}
Quadratic Goldreich-Levin algorithms were previously obtained by giving algorithmic proofs of the inverse theorem for the Gowers 3-uniformity norm ($U^3$-norm).
A basic fact that is used in these proofs is that the $U^3$-norm of a function is bounded from above by its sup-norm (its largest absolute value).
It follows from Young's inequality, however, that the stronger inequality $\|f\|_{U^3} \leq \|f\|_2$ holds true, where both norms are based on the uniform probability measure; moreover,~2 is the least value for which this inequality holds~\cite{eisner2012large}.

The extremizers of the $U^3$-norm relative to the $L_2$-norm turn out to have an important role in quantum information theory:
in that setting, they are known as \emph{stabilizer states}.
Quadratic phase functions are examples of stabilizer states, but this latter class is significantly richer, being furnished by the functions that can be written as a (non-classical) quadratic phase function restricted to an affine subspace of~$\F_2^n$.

Closely related to quadratic Goldreich-Levin algorithms is the \emph{stabilizer learning problem}, where one is given copies of an unknown quantum state that has correlation at least~$\tau$ with some stabilizer state and the goal is to find a stabilizer state with correlation at least~$\tau - \eps$.
Recent work of Chen, Gong, Ye and Zhang~\cite{chen2025stabilizer} gave computationally-efficient protocols for stabilizer learning with running time $\poly(n, 1/\eps) (1/\tau)^{O(\log(1/\tau))}$.
The first step of our algorithm is a suitable \emph{dequantization} of their quantum protocol.

Assuming that~$f$ is 1-bounded in sup-norm allows us to view it as a (possibly sub-normalized) quantum state, and having correlation at least~$\tau$ with a classical quadratic phase~$(-1)^{q(x)}$ makes it meet the right criteria for the input
(albeit with the different access form of \emph{classical queries} rather than \emph{quantum samples}).
This perspective allows us to draw from the ideas of~\cite{chen2025stabilizer}.
We give a slightly simplified version of their arguments when translated to our setting, so that no knowledge of quantum information theory or quantum computing is necessary to follow this paper.

A major source of difficulty when trying to implement the stabilizer learning algorithm of Chen et al.\ in our setting is that it is unclear how to implement their quantum sampling procedure (\emph{Bell difference sampling}) using few queries to the function~$f$.
Bell difference sampling can be done exactly using only~6 copies of a quantum state, and it is a crucial component of all known stabilizer learning algorithms.

A further discrepancy between the two settings is that the desired outputs in the stabilizer learning problem form a significantly larger set than the set of classical quadratic phases.
A general stabilizer state is an $L_2$-normalized non-classical quadratic phase that is supported on an affine subspace of arbitrary dimension.

To find a near-maximal quadratic correlator, our algorithm first finds a complete list of all stabilizer states that nearly attain the maximal correlation with~$f$ among all their ``neighbors.''
For each one of these approximate local maximizers of correlation, we use an algorithmic procedure to stretch its domain to the full space and remove its non-classical component.
With some care, one can show that one of the quadratic phase functions thus obtained is a near-optimal maximizer of correlation with~$f$.
Testing each of these correlations, we can find the near-optimal one.

Below, we give a short summary of our main algorithm with slightly more details.

\subsection{Outline of the algorithm}

An idea that goes back to the weak inverse theorem for the $U^3$-norm due to Gowers is to consider the large Fourier coefficients of the multiplicative derivatives of~$f$,
\begin{equation*}
    \widehat{\Delta_af}(b) = \Exp_{x\in \F_2^n}f(x+a)\overline{f(x)}(-1)^{ b\cdot x}.
\end{equation*}
It follows from Parseval's identity that the values 
\begin{equation*}
    P_f(a,b) = \frac{|\widehat{\Delta_af}(b)|^2}{2^n \|f\|_2^4}
\end{equation*} 
form a probability distribution over~$\F_2^n\times\F_2^n$.

Recent work on stabilizer testing, which is closely related to a $U^3$-inverse theorem relative to the~$L_2$-norm, highlights the role of isotropy in this context~\cite{arunachalam2025polynomial, bao2025tolerant}.
A subset~$V\subseteq \F_2^n\times\F_2^n$ is isotropic if all pairs $(a,b),(c,d)\in V$ have zero symplectic inner product:
\begin{equation*}
    [(a,b), (c,d)] =  a\cdot d + b\cdot c.
\end{equation*}
A slight variation of a standard integration argument from the $U^3$-inverse theorem shows that, if~$V$ is an isotropic subspace, then there exists a stabilizer state~$\phi$ with correlation $|\langle f,\phi\rangle|^2 \geq P_f(V)$~\cite{vanDordrecht2025}.

A standard fact from quantum computing is that the quadratic component of a stabilizer state is uniquely characterized by a 
\emph{Lagrangian subspace}: an isotropic subspace of maximal dimension~$n$.
This means that our goal of finding a correlating stabilizer state~$\phi$ may be shifted to finding its associated Lagrangian subspace~$L$.
Once the quadratic component has been found, the Goldreich-Levin algorithm gives a means to learn its  linear component.

Parseval's identity and the Cauchy-Schwarz inequality show that sampling from~$P_f$ is likely to yield elements from~$L$, since they give that
\begin{equation*}
    P_f(L) \geq |\langle f,\phi\rangle|^4.
\end{equation*}
A complication is that, since we do not know~$\phi$, we have no way to certify that a given~$P_f$ sample belongs to~$L$.

A key insight of~\cite{chen2025stabilizer} is that a particular uncertainty principle shows that the \emph{spectral set} --- the set of elements~$(a,b)$ for which $P_f(a,b)$ exceeds $1/2$  --- is isotropic.
In the 100\% setting, where~$f$ is itself a stabilizer state, the isotropic set equals~$L$ and~$P_f$ is the uniform probability distribution over~$L$.
One important aspect of this is that the Goldreich-Levin algorithm can be used to check  to a high degree of certainty whether a given $(a,b)$ pair belongs to the spectral set.

A problem that arises now is how to efficiently generate the spectral set.
Intuition from the proof of the $U^3$-inverse theorem suggests that this set may have some linear structure resembling a subspace.\footnote{This would follow from the Balog-Szemer\'edi-Gowers theorem and the Fre\u{\i}man-Ruzsa theorem if the spectral set had positive~$P_f$-measure, but this needs not be the case;  these details are not important, however, as they will not be explicitly used here.}
In this setting, approximate linear structure appears more straightforwardly in the following form. 
Consider a random set~$F\subseteq \F_2^n\times\F_2^n$ of~$\Theta(n)$ elements sampled independently from the conditional distribution~$P_f$, normalized over the spectral set.
Then with good probability,~$\vspan(F)$ will cover all but a tiny fraction (more precisely, $P_f$-measure) of the whole spectral set.
In this case, we refer to~$\vspan(F)$ as an approximate spectral space.

If the correlation between~$f$ and~$\phi$ is sufficiently close to~1, then the approximate spectral space equals~$L$ with good probability.
In general, however, this probability can be arbitrarily small even if the correlation is bounded away from zero.
This can be dealt with using the following observation from~\cite{chen2025stabilizer}:
If $(a,b)\in L$, then $\Delta_a\phi(x)=\sigma(-1)^{ b\cdot x}$ for some sign~$\sigma\in \pmset{}$.
If moreover~$(a,b)$ is not in the spectral set, then by projecting~$f$ to the subspace of functions satisfying this same identity, we obtain a new function~$f'$ with a correlation with~$\phi$ that is larger by a factor of at least~$1.08$.
Since the maximum correlation is capped by~1, we can only iterate this a bounded number of times.
If~$\phi$ is a local maximizer of correlation for~$f$, then within these iterations we will necessarily have generated a projected version of~$f$ from which we can obtain the desired Lagrangian~$\Lcal(\phi)$ with good probability by generating an approximate spectral space.
The fact that this works with constant probability immediately implies that there can only be a constant number of local maximizers of correlation.
Hence, by repeating only a constant number of times, we can with high probability obtain the entirely list of such correlators.

From this list of stabilizer states, we wish to obtain a list of correlating quadratic phase functions.
Since the input function~$f$ is assumed to be bounded, it follows that any stabilizer state with significant correlation with~$f$ must be supported on an affine subspace of constant codimension.
This implies that, for each correlating stabilizer state~$\phi$, we can find a constant-sized list of full-support stabilizer states containing at least one element~$\psi$ having high correlation with~$f$.
Having thus ``streched'' their domain, we next replace the non-classical quadratic component of each full-support stabilizer state with a classical quadratic component so as to obtain our list of classical quadratic phases.
This is done using the fact that any non-classical quadratic phase function will behave classically in some hyperplane, and incurs in only a slight loss in correlation
We show that, if we choose the parameters correctly, then the highest-correlation quadratic phase in this new list will have correlation at least $\|f\|_{u^3} - \eps$.

An obstacle to implementing the above algorithm in our query model is that we do not have direct access to samples from the distribution~$P_f$.
However, since we need only~$\poly(n)$ samples from~$P_f$ and the analysis of our algorithm only deals with isotropic subsets of~$\F_2^n\times\F_2^n$ --- which are bounded in size by~$2^n$ --- sampling from a distribution that only crudely approximates~$P_f$ turns out to be sufficient for our purposes.
Our approximating distribution is obtained by sampling $a\in\F_2^n$ uniformly and using the Goldreich-Levin algorithm to sample $b\in \F_2^n$ with probability close to $|\widehat{\Delta_af}(b)|^2$.
With high probability, in most of the space~$\F_2^{2n}$, this distribution is close to a fixed multiple of~$P_f$.

\section{Preliminaries}

Given vectors $a, b\in \F_2^n$, define their \emph{inner product} by
$$a\cdot b = a^{\mathsf T}b = a_1b_1 + \dots+ a_nb_n$$
and their \emph{entry-wise product} by
$$a\circ b = (a_1b_1,\, a_2 b_2,\, \dots,\, a_nb_n).$$
For functions $f, g:\F_2^n \to \C$, denote $\langle f, g\rangle = \E_{x\in \F_2^n} f(x) \overline{g(x)}$ and $\|f\|_2 = \langle f, f\rangle^{1/2}$.
The support of a function~$f$ is denoted by $\supp(f)$.

Recall that $\|f\|_{u^3} = \max \big|\Exp_{x\in \F_2^n}f(x)(-1)^{q(x)}\big|$, where the maximum is taken over all polynomials $q: \F_2^n \to \F_2$ of degree at most~2.
We say that~$f$ is \emph{1-bounded} if $|f(x)|\leq 1$ for all~$x\in \F_2^n$.
Denote $\D = \{z\in\C:\: |z|\leq 1\}$ and $S^1 = \{z\in\C:\: |z|=1\}$.

\subsection{Algorithms for linear Fourier analysis}

We use the following version of the Goldreich-Levin algorithm, which is a special case of~\cite[Theorem~4.3]{kim2023cubic}.

\begin{theorem}[Goldreich-Levin algorithm] \label{thm:GL}
    Let $f:\F_2^n\to \C$ be a 1-bounded function, let~$\delta>0$ and $0<\tau\leq 1$. 
    There is a randomized algorithm that makes $n \log n\, \poly(\log(1/\delta)/\tau)$ queries to~$f$ and, with probability at least~$1-\delta$, outputs a list $L\subseteq \F_2^n$ such that:
    \begin{itemize}
        \item If $|\widehat{f}(b)| \geq \tau$, then $b\in L$;
        \item For every $b\in L$, we have $|\widehat{f}(b)| \geq \tau/2$.
    \end{itemize}
    The running time of this algorithm is $n \log n\, \poly(\log(1/\delta)/\tau)$.
\end{theorem}

Additionally, we will need the following standard concentration inequality.

\begin{lemma}\label{lem:hoeffding}
Let $X_1,\dots, X_n$ be independent $\C$-valued random variables such that $|X_i| \leq a_i$ for some~$a_i>0$.
Let $\overline X = n^{-1}(X_1 + \cdots + X_n)$.
Then, for any $\eps > 0$,
\begin{equation*}
    \Pr\big[\big|\overline{X} - \Exp\overline{X}\big| > \eps\big] \leq 4\exp\Bigg(-\frac{2\eps^2n^2}{\sum_{i=1}^na_i^2}\Bigg).
\end{equation*}
\end{lemma}

\begin{proof}
    Let $Y_i^0$ and $Y_i^1$ be the real and complex parts of~$X_i$, respectively.
    Then, $|Y_i^0|,|Y_i^1|\leq a_i$.
    For $b\in \bset{}$, let $\overline{Y^b} = n^{-1}(Y_1^b + \cdots + Y_n^b)$.
    By Hoeffding's inequality~\cite{Hoeffding:1963}, for $b\in \bset{}$,
    \begin{equation*}
        \Pr\big[\big|\overline{Y^b}  - \Exp\overline {Y^b}\big|>\eps\big] \leq 2\exp\Bigg(-\frac{2\eps^2n^2}{\sum_{i=1}^na_i^2}\Bigg).
    \end{equation*}
    The result now follows from the triangle inequality and the union bound.
\end{proof}

The next lemma allows us to estimate the Fourier coefficients of a given bounded function:

\begin{lemma}\label{lem:fourest}
    Let $\eps,\delta>0$.
    There exists a randomized algorithm $\FourEst_{\eps,\delta}$ that, given $b\in \F_2^n$ and query access to an $1$-bounded function $f:\F_2^n\to \C$, makes $O(\frac{1}{\eps^2}\log(1/\delta))$ queries to~$f$ and, with probability at least $1- \delta$, returns a value $c\in \C$ such that $|c - \widehat{f}(b)|\leq \eps$.
\end{lemma}

\begin{proof}
    Let $m = O(\frac{1}{\eps^2}\log(1/\delta))$, let $x_1,\dots,x_m$ be independent uniformly distributed $\F_2^n$-valued random variables and let $X_i = f(x_i)(-1)^{\langle x_i,b\rangle}$ for each $i\in[m]$.
    Then $\Exp[X_i] = \widehat{f}(b)$ for each $i\in [m]$.
    It follows from Lemma~\ref{lem:hoeffding} that~$c = m^{-1}(X_1,\dots,X_m)$ satisfies the requirement with the desired probability.
\end{proof}

\subsection{Gowers uniformity in~$L_2$ and stabilizer states}

\begin{definition}[Gowers $U^3$-norm]
    Given a function $f:\F_2^n\to\C$ and $a\in \F_2^n$, let $\Delta_af(x) = f(x+a)\overline{f(x)}$.
    The $U^3$-norm of~$f$ is then defined by
    \begin{equation*}
        \|f\|_{U^3} = \big(\Exp_{x,a,b,c \in \F_2^n}\Delta_a\Delta_b\Delta_cf(x)\big)^{\frac{1}{8}}.
    \end{equation*}
\end{definition}

It is easy to see that $\|f\|_{U^3} \leq 1$ for every 1-bounded function~$f$.
The functions that attain the equality are called \emph{non-classical quadratic phase functions}.

\begin{definition}[Non-classical quadratic phase functions]
    For a linear vector space~$V$ over~$\F_2$, a function $\psi: V\to \C$ is a non-classical quadratic phase function if, for any $x, a, b, c\in V$, we have that
    \begin{equation*}
        \Delta_a \Delta_b \Delta_c \psi(x) = 1.
    \end{equation*}
\end{definition}

Define $|\cdot|:\F_2\to\{0,1\}$ to be the natural identification map.
This map satisfies the identity $|a +b| = |a| + |b| - 2|ab|$.
With slight abuse of notation, for $x\in \F_2^n$, define $|x| = |x_1| + \cdots |x_n|$.
It turns out that every non-classical quadratic phase function on~$\F_2^n$ can be written as
$$\psi(x) = \alpha (-1)^{x^{\mathsf T} Ax + b\cdot x} i^{|c\circ x|}$$
for some matrix $M\in \F_2^{n\times n}$, vectors $b, c \in \F_2^n$ and scalar~$\alpha\in S^1$ \cite{TaoZiegler2012}.

By Young's inequality, one can show that $\|f\|_{U^3} \leq 1$ holds for all functions~$f$ that have $L_2$-norm~$1$.
The functions that attain the equality are known as \emph{stabilizer states}:

\begin{definition}[Stabilizer states]
A function $\phi:\F_2^n\to \C$ is a \emph{stabilizer state} if it satisfies $\|\phi\|_{2} = \|\phi\|_{U^3} = 1$.
Denote the set of stabilizer states by~$\Stab(\F_2^n)$.
\end{definition}

The next result was obtained in~\cite{eisner2012large}:

\begin{proposition}\label{prop:stab_states_explicit}
    A function $\phi: \F_2^n \to \C$ is a stabilizer state if and only if there exists a subspace~$V\subseteq \F_2^n$, a vector $u\in \F_2^n$ and a non-classical quadratic phase function $\psi: V\to \C$ such that
    $$\phi(x) = 2^{(n-\dim(V))/2} \one_{V}(x+u) \psi(x+u) \quad \text{for all $x\in \F_2^n$.}$$
\end{proposition}

By the classification of non-classical quadratic phase functions, we can write a stabilizer state explicitly as a function of the form
\begin{equation}\label{eq:stabstate}
\phi(x) = \alpha 2^{(n-\dim(V))/2} \one_{u+V}(x) (-1)^{q(x)} i^{|c\circ x|}, \end{equation} 
where $\alpha\in S^1$, $V \subseteq \F_2^n$ is a subspace, $q: \F_2^n \to \F_2$ is a quadratic function, $u \in \F_2^n$ and either~$c=0$ or $c\notin V^{\perp}$.
We will usually ignore the global phase~$\alpha$, as it makes no contribution to the correlation $|\langle f, \phi\rangle|$.

The following definition from~\cite{chen2025stabilizer} will be crucial for our arguments:

\begin{definition}[Approximate local maximizer]
    Let~$f:\F_2^n \to \C$ be some function and $\gamma>0$ be a positive parameter.
    A stabilizer state $\phi\in \Stab(\F_2^n)$ is a \emph{$\gamma$-approximate local maximizer of correlation for~$f$} if it satisfies
    \begin{equation*}
        |\langle f,\phi\rangle|^2
        \geq
        \gamma \max_{\phi'\in \Stab(\F_2^n),\, |\langle\phi,\phi'\rangle|^2 = \frac{1}{2}}|\langle f,\phi'\rangle|^2.
    \end{equation*}
\end{definition}

\subsection{Symplectic geometry}

\begin{definition}[Symplectic inner product]\label{def:sympspace}
For vectors $(x,y), (x',y')\in \F_2^n \times \F_2^n$, define their \emph{symplectic inner product} by
    \begin{equation*}
        [(x,y),(x',y')] = \langle x,y'\rangle + \langle x',y\rangle.
    \end{equation*}
\end{definition}

\begin{definition}[Isotropic subspace]
    A subspace $V\subseteq \F_2^n \times \F_2^n$ is \emph{isotropic} if it holds that $[u,v] =0$ for all $u,v\in V$.
    A subspace~$L$ is \emph{Lagrangian} if it is isotropic and satisfies $\dim(L) = n$.
\end{definition}

One can easily show that a subspace~$L$ is Lagrangian if and only if it can be written as
$$L = \big\{(h,\, Mh+w):\: h\in V,\, w\in V^{\perp}\big\}$$
for some subspace $V\leq \F_2^n$ and some symmetric matrix $M\in \F_2^{n\times n}$.
The importance of these notions for us is that every stabilizer state is naturally associated with a Lagrangian subspace:

\begin{definition}[Lagrangian subspace] \label{def:Lagrangian_phi}
    Given a stabilizer state~$\phi\in \Stab(\F_2^n)$, define
    $$\Lcal(\phi) = \big\{(a, b) \in \F_2^n:\: |\widehat{\Delta_a \phi}(b)| = 1\big\}.$$
\end{definition}

This set~$\Lcal(\phi)$ is a Lagrangian subspace.
Explicitly, for a stabilizer state as in~\eqref{eq:stabstate}, where $q(x) = x^{\mathsf T}(Ax+b)$ for some matrix $A\in \F_2^{n\times n}$ and vector $b\in \F_2^n$, its associated Lagrangian has the form
\begin{equation}\label{eq:Lphi}
    \mathcal L(\phi)
    =
    \big\{\big(a, \big(A^\mathsf T + A + \Diag(c)\big)a + w\big)\mid a\in V, w\in V^\perp\big\}.
\end{equation}

We use the following useful result from~\cite[Lemma~4.7]{chen2025stabilizer}.

\begin{lemma}[Uncertainty principle]\label{lem:uncertainty1}
    Let $f:\F_2^n\to \C$ be a function and let $(a,b), (c,d)\in \F_2^n\times\F_2^n$ be pairs satisfying $[(a,b), (c,d)] = 1$.
    Then
    \begin{equation*}
       |\widehat{\Delta_af}(b)|^2 + |\widehat{\Delta_{c}f}(d)|^2 \leq \|f\|_2^4.
    \end{equation*}
\end{lemma}

As an immediate consequence of the uncertainty principle, we see that the set
\begin{equation*}
    \big\{(a,b)\in \F_2^n\times\F_2^n \mid |\widehat{\Delta_af}(b)|^2 > \tfrac{1}{2} \|f\|_2^4\big\}
\end{equation*}
is isotropic.

\subsection{Probability distributions and Weyl operators}

\begin{definition}[Characteristic and convoluted distributions]
    For a nonzero function $f:\F_2^n\to \C$, define the \emph{characteristic} and \emph{convoluted} distributions of~$f$ respectively as
    \begin{align*}
        P_f(a,b) &= \frac{1}{2^n \|f\|_2^4}|\widehat{\Delta_af}(b)|^2\\
        Q_f(a,b) &= (P_f*P_f)(a,b).
    \end{align*}
\end{definition}

If~$\phi$ is a stabilizer state, one can show that~$P_\phi$ and~$Q_\phi$ are both equal to the uniform probability distribution over the Lagrangian~$\Lcal(\phi)$.

If~$f$ gives the computational-basis description of a quantum state, sampling from the convoluted distribution~$Q_f$ is known as \emph{Bell difference sampling}.
In the case where~$f$ correlates with some stabilizer state~$\phi$, we have that $\Lcal(\phi)$ has large mass according to~$P_f$ and~$Q_f$:

\begin{lemma}\label{lem:Lcorr}
Let~$\phi:\F_2^n\to\C$ be a stabilizer state and let~$L = \Lcal(\phi)$ be its Lagrangian subspace as in Definition~\ref{def:Lagrangian_phi}.
Then, for any function $f:\F_2^n\to \C$, we have that, 
\begin{equation*}
    Q_f(L) \geq P_f(L)^2 \geq |\langle f,\phi\rangle|^8.
\end{equation*}
\end{lemma}

\begin{proof}
The first inequality follows easily because~$L$ is a linear subspace.
For the second inequality, we may assume that~$\|f\|_2 = 1$.
It then follows from the Cauchy-Schwarz inequality and Parseval's identity that 
    \begin{align*}
    |\langle f, \phi\rangle|^2 &= \Exp_{a\in\F_2^n}\langle \Delta_af,\Delta_a\phi\rangle\\
    &=\Exp_{a\in \F_2^n}\sum_{b\in \F_2^n}\widehat{\Delta_af}(b) \overline{\widehat{\Delta_a\phi}(b)}\\
    &=\Exp_{(a,b)\in L}\widehat{\Delta_af}(b) \overline{\widehat{\Delta_a\phi}(b)}\\
    &\leq \Big(\Exp_{(a,b)\in L}\big|\widehat{\Delta_af}(b)\big|^2\Big)^{\frac{1}{2}}\\
    &=P_f(L)^{\frac{1}{2}}.
\end{align*}
This proves the claim.
\end{proof}

\begin{definition}[Weyl operators]\label{def:Weyl}
    For a pair $(a,b)\in \F^n\times\F_2^n$ define the linear operator $W_{a,b}$ on functions $f:\F_2^n\to \C$ by
    \begin{equation*}
        (W_{a,b}f)(x) = i^{|a\circ b|}(-1)^{b\cdot x}f(x+a).
    \end{equation*}
\end{definition}

While we will not use the Weyl operators much here, we record a few basic facts about them to allow us to defer the proofs of some of the facts we use below to quantum-information-theoretic literature (see for instance~\cite{Gross2021}).

\begin{itemize}
    \item The Weyl operators are unitary and Hermitian, and so have eigenvalues in~$\pmset{}$.
    \item $W_{a,b}W_{c,d} = (-1)^{[(a,b),(c,d)]}W_{c,d}W_{a,b}$ for all $a,b,c,d\in \F_2^n$.
    \item $\widehat{\Delta_af}(b) = i^{|a\circ b|}\langle f,W_{a,b}f\rangle$ for all $a,b\in \F_2^n$.
\end{itemize}

\section{The main algorithm}

In this section we provide an algorithm that, when given query access to a bounded function $f: \F_2^n \to \C$ and parameters $\tau>0$, $1/2 <\gamma \leq 1$, returns a list of size $\big((\gamma-1/2) \tau\big)^{O(-\log(1/\tau))}$ containing \emph{all} $\gamma$-approximate local maximizers $\phi\in \Stab(\F_2^n)$ satisfying $|\langle f, \phi\rangle| \geq \tau$
(with high probability).
This can be regarded as a ``dequantization'' of the quantum procedure given by~\cite[Corollary~6.2]{chen2025stabilizer}, as well as a type of list-decoding algorithm (which is only possible due to the notion of approximate local maximality).
In the next section we will show how to use this algorithm to prove the results stated in the Introduction.

\begin{theorem}[List-decoding stabilizer states] \label{thm:list}
    Let $\tau,\, \delta>0$ and $1/2<\gamma\leq 1$.
    There is a randomized algorithm that, when given query access to a bounded function $f: \F_2^n \to [-1, 1]$, returns a list of size $\log(1/\delta) \big((\gamma-1/2) \tau\big)^{-O(\log(1/\tau))}$ which, with probability at least $1-\delta$, contains all stabilizer states that are $\gamma$-approximate local maximizers and have correlation at least~$\tau$ with~$f$.
    This algorithm makes $n^2 \log n\, \log(\delta^{-1}) \big((\gamma-1/2) \tau\big)^{-O(\log(1/\tau))}$ queries to~$f$ and has runtime $n^3 \log(\delta^{-1}) \big((\gamma-1/2) \tau\big)^{-O(\log(1/\tau))}$.
\end{theorem}

We will start by providing a sampling algorithm that, with nonnegligible probability, outputs the Lagrangian subspace $\Lcal(\phi)$ associated to some fixed (but unknown) $\gamma$-approximate local maximizer of correlation~$\phi$ satisfying $|\langle f, \phi\rangle| \geq \tau$.
This algorithm assumes both query access to the function~$f$ and sampling access to its convoluted distribution~$Q_f$
(as well as~$Q_{f'}$ for any function~$f'$ that is ``easily queried'' from~$f$).
This is given in Section~\ref{sec:Lagrangian}, and essentially amounts to a simplified version of the arguments of Chen et al.\ when translated to our setting.

We then show, in Section~\ref{sec:Stabilizer}, how to learn (with nonnegligible probability) the desired stabilizer state~$\phi$ from its Lagrangian~$\Lcal(\phi)$.
Note that there are~$2^n$ stabilizer states associated to any Lagrangian subspace, and several of them can satisfy the requirements of our unknown stabilizer state~$\phi$.
Our learning algorithm will then output a random such stabilizer state~$\psi$ whose probability of being picked depends only on the correlation $|\langle f, \psi\rangle|$.

In Section~\ref{sec:Convoluted} we provide a randomized sampling procedure which, when given the ability to query $O(n \log n)$ times a bounded function $f: \F_2^n \to \C$, samples from some probability distribution over $\F_2^{2n}$ that is ``close enough'' to the convoluted distribution~$Q_f$ (albeit in a nontrivial way).
This will allow us to ``transform'' samples from~$Q_f$ into queries to~$f$ at a cost of $O(n\log n)$ queries per sample.

Finally, in Section~\ref{sec:List} we knit all the results obtained thus far into the ``list-decoding'' algorithm we want.

\subsection{Sampling a good Lagrangian subspace} \label{sec:Lagrangian}

Recall that, as a consequence of the uncertainty principle given in Lemma~\ref{lem:uncertainty1}, the set
\beqn
\big\{(a,b)\in \F_2^n\times\F_2^n \mid |\widehat{\Delta_af}(b)|^2 > \tfrac{1}{2} \|f\|_2^4\big\}
\eeqn
is isotropic.
In order to account for possible approximation errors, we will consider the following set:

\begin{definition}[Spectral set]
    For a function $f:\F_2^n\to \C$, define
    \begin{equation*}
        \Spec(f) = \big\{(a,b)\in \F_2^n\times\F_2^n \mid |\widehat{\Delta_af}(b)|^2 \geq 0.7 \|f\|_2^4\big\}.
    \end{equation*}
\end{definition}

The spectral set is then clearly isotropic as well.

The distributions~$P_f$ and~$Q_f$ are biased towards elements in $\Spec(f)$.
Moreover, given a pair $(a,b)\in \F_2^{2n}$, we can easily estimate the value of $|\widehat{\Delta_af}(b)|^2$ using $O(1)$ queries to~$f$, and thus we have an approximate membership oracle for the set $\Spec(f)$.
These two facts combined make dealing with $\Spec(f)$ very useful.

\subsubsection{Robust Lagrangian generation}

\begin{definition}
    Let $f:\F_2^n\to \C$ be a nonzero function.
    A set $F\subseteq \F_2^{2n}$ is an $\eps$-approximate spectral set for~$f$ if 
    \begin{equation*}
        Q_f\big(\Spec(f)\setminus F\big) \leq \eps.
    \end{equation*}
\end{definition}

\begin{definition}
    Let~$f:\F_2^n\to \C$ be a nonzero function, let $L \subseteq \F_2^n \times \F_2^n$ be a Lagrangian subspace and let $0<\eta<1$.
    We say that~$f$ \emph{$\eta$-robustly generates~$L$} if $L \subseteq \vspan(F)$ for every $\eta$-approximate spectral set~$F$.
\end{definition}

If~$f$ robustly generates~$\Lcal(\phi)$, then it is easy to learn a basis of~$\Lcal(\phi)$ by sampling $O(n/\eta)$ pairs $(a, b) \sim Q_f$.
This is because the span of such a sample, after pruning, is an approximate spectral set with good probability.

\begin{lemma}\label{lem:rubustgen}
    Let $\eps,\delta>0$ and $f:\F_2^n\to \C$ be a $1$-bounded function.
    Suppose that~$f$ $\eps$-robustly generates a Lagrangian subspace~$L$.
    There is a randomized algorithm that uses $m = O\big(\tfrac{1}{\eps}(n + \log(\tfrac{1}{\delta})\big)$ samples from~$Q_f$, makes $O(m\log(m/\delta))$ queries to~$f$ and returns a basis for a subspace~$L'\subseteq \F_2^n\times\F_2^n$ such that with probability at least~$1-\delta$, we have~$L'=L$.
\end{lemma}

The proof of this lemma is essentially given in~\cite{chen2025stabilizer}.
We sketch it here for completeness.

\begin{proof}[ sketch]
  Let $S\subseteq \F_2\times\F_2^n$ be a random set of~$m$ independent $Q_f$-samples and let $T = S\cap\Spec(f)$.
  We first show that with probability at least $1 - \delta/2$, the set~$T$ is an $\eps$-approximate spectral set.
  
  Let $p = Q_f(\Spec(f))$.
  If~$p\leq \eps$, then there is nothing to prove.
  Suppose $p > \eps$.
  Note that the elements of~$T$ are distributed independently according to~$R_f=\one_{\spec(f)}Q_f/p$.
  By the Chernoff bound, we have that $|T| \geq (pm)/2$ with probability at least $1 - \delta/4$.
  Conditioned on this size of~$T$, it follows from~\cite[Lemma~2.3]{grewal2023efficient} that with probability at least~$1 - \delta/4$, we have $R_f(\vspan(T)) \geq 1-\eps/p$.
  Then,
  \begin{align*}
      Q_f\big(\spec(f)\setminus \vspan(T)\big) &=
      p\,R_f\big(\spec(f)\setminus \vspan(T)\big)
      \leq
      \eps.
  \end{align*}
  This shows that with probability at least $1 - \delta/2$, the set~$T$ is an $\eps$-approximate spectral set.

  For each $(a,b)\in S$, run the algorithm $\FourEst_{\eps_1,\delta_1}$ from Lemma~\ref{lem:fourest} on input $(\Delta_af, b)$ with parameters $\eps_1 = 0.1$ and $\delta_1 = \delta/(4m)$.
  Let~$F\subseteq S$ be the set for which the algorithm returns a complex number~$c$ such that~$|c|>0.6$.
  By the union bound, with probability at least $1 - \delta/4$, we have that $T\subseteq F$.
  Hence~$F$ is an $\eps$-spectral set with probability at least $1 - 3\delta/4$.
  Moreover,~$F$ contains no elements such that~$|\widehat{\Delta_af}(b)| \leq 0.5$ with probability at least~$1 - \delta/4$.
  In this case,~$F$ is isotropic by Lemma~\ref{lem:uncertainty1}.
  Since~$f$ $\eps$-robustly generates~$L$, we get that~$\vspan(F) = L$ with probability at least~$1-\delta$.

  Return a basis for~$F$.
\end{proof}

\subsubsection{Non-robust Lagrangian generation implies energy increment.}

If~$f$ \emph{does not} generate~$\mathcal L(\phi)$ robustly, then there is an easy way to obtain an ``energy increment'' given by an increase of the normalized correlation with~$\phi$.
This is obtained by replacing~$f$ with a projection of~$f$ to a subspace of functions satisfying a certain linear relation satisfied by~$\phi$.
Since~$\phi$ is a stabilizer state, for every choice of~$a$, it satisfies that its discrete derivative~$\Delta_a\phi$ is proportional to a linear phase function on a coset of a subspace.

Given $a, b\in \F_2^n$ and $\sigma\in \pmset{}$, 
define the subspace of functions
\begin{equation*}
    V_{a,b}^{\sigma} = \big\{f: \F_2^n\to \C\mid f(x+a) = \sigma i^{-|a\circ b|} (-1)^{b\cdot x} f(x)\:\:\text{for all $x\in \F_2^n$}\big\}.
\end{equation*}
It follows readily from the explicit forms of a stabilizer state~$\phi$ given in~\eqref{eq:stabstate} and its associated Lagrangian~$\mathcal L(\phi)$ given in~\eqref{eq:Lphi} that for each $(a,b)\in \mathcal L(\phi)$ there is  a $\sigma$ such that $\phi\in V_{a,b}^\sigma$.

The projection $\Pi_{a,b}^{\sigma}f$ of a function~$f$ to~$V_{a,b}^{\sigma}$ is given by the function
\begin{equation*}
    \Pi_{a,b}^{\sigma}f(x) = \frac{f(x) + \sigma i^{|a\circ b|}(-1)^{b\cdot x} f(x+a)}{2}.
\end{equation*}

\begin{lemma}[Energy boosting]\label{lem:boost}
    Let~$f:\F_2^n\to\D$ be a 1-bounded function and suppose~$\phi\in V_{a,b}^{\sigma}$ is a stabilizer state.
    If $(a,b)\not\in \Spec(f)$, then the function
    \begin{equation*}
        f' = \Pi_{a,b}^{\sigma}f
    \end{equation*}
    satisfies 
    \begin{equation*}
        \frac{|\langle f',\phi\rangle|^2 }{\|f'\|_{2}^2} \geq 1.08\frac{|\langle f,\phi\rangle|^2}{\|f\|_{2}^2}.
    \end{equation*}
    Moreover, if~$\phi$ is a $\gamma$-approximate local maximizer for~$f$, then it is also a $\gamma$-approximate local maximizer for~$f'$.
\end{lemma}

\begin{proof}
Since $\phi\in V_{a,b}^\sigma$, we have that $\Pi_{a,b}^\sigma\phi = \phi$, and so
\begin{equation*}
    \langle f',\phi\rangle = \langle \Pi_{a,b}^\sigma f,\phi\rangle = \langle f,\Pi_{a,b}^\sigma\phi\rangle = \langle f,\phi\rangle. 
\end{equation*}
We also have that
\begin{align*}
    \|f'\|_2^2 &\leq \frac{1}{2}\|f\|_2^2 + \frac{1}{2}|\widehat{\Delta_af}(a)| \leq \frac{1}{2}\big(1 + \sqrt{0.7}\big) \|f\|_2^2 \leq 0.92 \|f\|_2^2.
\end{align*}
This implies the first claim.

Recall the definition of the Weyl operators (Definition~\ref{def:Weyl}).
Suppose~$\phi$ is a $\gamma$-approximate local maximizer for~$f$.
Any $\phi'\in\Stab(\F_2^n)$ satisfying $|\langle \phi,\phi'\rangle| = 1/\sqrt{2}$ has the form $\tfrac{1}{\sqrt{2}}(I + i^\ell W_{c,d})\phi$ for some $\ell\in \Z$ and $c,d\in \F_2^n$~\cite[Theorem~13]{Garcia2014Geometry}.
Since $\phi\in V_{a,b}^\sigma$ it follows that $\langle f',\phi\rangle = \langle f,\phi\rangle$.

Now let $\phi' = \tfrac{1}{\sqrt{2}}(I + i^\ell W_{c,d})\phi$.
Let $M = \tfrac{1}{\sqrt{2}}(I + i^\ell W_{c,d})$.
If $[(a,b),(c,d)] = 0$, then $\Pi_{a,b}^\sigma$ and $M$ commute and we get that
Then,
\begin{align*}
    \langle f',\phi'\rangle = \langle f, \Pi_{a,b}^\sigma M\phi\rangle
    =\langle f,\phi'\rangle.
\end{align*}
This gives
\begin{align*}
    |\langle f',\phi'\rangle|^2 &=|\langle f,\phi'\rangle|^2
    \leq \frac{1}{\gamma}|\langle f,\phi\rangle|^2
    = \frac{1}{\gamma}|\langle f',\phi\rangle|^2.
\end{align*}

If $[(a,b),(c,d)] = 1$, then $\Pi_{a,b}^\sigma M\phi = \frac{1}{\sqrt{2}}\phi$ and so $\langle f',\phi'\rangle = \frac{1}{\sqrt{2}}\langle f',\phi\rangle$.
Since~$\gamma\leq 1$, this implies that
\begin{equation*}
    |\langle f',\phi'\rangle|^2 \leq \frac{1}{\gamma}|\langle f',\phi\rangle|^2.
\end{equation*}
This proves the claim.
\end{proof}

The idea now is to iteratively use Lemma~\ref{lem:boost} until a function has been found that robustly generates~$\mathcal L(\phi)$, at which point the algorithm from Lemma~\ref{lem:rubustgen} can be used to find~$\mathcal L(\phi)$ with good probability.
The main observation to make is that if~$|\langle f,\phi\rangle|\geq \tau$, then the energy can be boosted at most~$t= O(\log(1/\tau))$ times until we have obtained a projection of~$f$ that perfectly correlates with~$\phi$ and thus maximizes the energy. 
Hence, if we choose $t'$ uniformly at random from~$\{0,\dots,t\}$ and boost~$t'$ times, with probability at least~$1/t$ we will have obtained a projection of~$f$ that robustly generates~$\mathcal L(\phi)$.

It turns out that if~$f$ does not robustly generate~$\Lcal(\phi)$, then it is not hard to find a projection as in Lemma~\ref{lem:boost}. 
The following lemma shows that in this case, a sample from~$Q_f$ will with non-negligible probability yield a pair $(a,b)\in \mathcal L(\phi)\setminus \spec(f)$.
Flipping a coin to choose a sign~$\sigma$ then gives a triple $(a,b,\sigma)$ enabling an energy boost with good probability.

\begin{lemma}\label{lem:nonrobust}
    Let~$\gamma\in (\tfrac{1}{2}, 1)$, $\tau>0$ and denote $\eps = (\gamma  -\frac{1}{2})^2\tau^8/8$.
    Let~$f:\F_2^n\to\C$ be a function and let~$\phi$ be a $\gamma$-approximate local maximizer of correlation for~$f$ such that $|\langle f,\phi\rangle| \geq \tau$.
    Suppose that~$f$ does not $\eps$-robustly generate~$\mathcal L(\phi)$.
    Then,
    \begin{equation*}
        Q_f\big(\mathcal L(\phi)\setminus \spec(f)\big) \geq \eps.
    \end{equation*}
\end{lemma}

The proof of Lemma~\ref{lem:nonrobust} uses that in the non-robust setting, there is an approximate spectral set~$F$ such that $\Lcal(\phi) \cap\vspan(F)$ is a strict subspace of~$\Lcal(\phi)$ and the fact that the convoluted distribution~$Q_f$ is smoothly distributed over the cosets of strict subspaces of~$\mathcal L(\phi)$ if~$\phi$ is an approximate local maximizer of correlation for~$f$.
(This is where using~$P_f$ would not work.)
This is proved in the lemmas below.

\begin{lemma} \label{lem:terrible}
    Let $f:\F_2^n\to \C$ be a function and suppose that $\phi = 2^{n/2}\one_{\{0\}}$ is a $\gamma$-approximate local maximizer for~$f$.
    Let $V\subseteq \F_2^n$ be a subspace with codimension~1.
    Then,
    \begin{equation*}
        Q_f(\{0^n\}\times (\F_2^n\setminus V)) \geq \tfrac{1}{4} \big(\gamma - \tfrac{1}{2}\big)^2|\langle f,\phi\rangle|^8.
    \end{equation*}
\end{lemma}

\begin{proof}
Let~$u\in \F_2^n\setminus\{0\}$ be such that~$V = \{u\}^\perp$.
  We begin by showing that
  \begin{equation}\label{eq:gammaballs}
      y:= \frac{f(u)^2}{f(0)^2}\not\in \{-1,1\} + \big(\gamma - \tfrac{1}{2}\big)\D.
  \end{equation}
  Indeed, since~$\phi$ is a $\gamma$-approximate local maximizer of correlation for~$f$, it follows that
  \begin{equation*}
      2^{-n}|f(0)|^2 \geq \gamma\max_{a\in \Z_4}\big|\langle f, 2^{(n-1)/2}(\one_{\{0\}} + i^a\one_{\{u\}})\rangle\big|^2.
  \end{equation*}
  In turn, this implies that
  \begin{equation*}
      \frac{f(u)}{f(0)} \not\in \{1,i,-1,-i\} + \big(\gamma - \tfrac{1}{2}\big)\D,
  \end{equation*}
  which gives~\eqref{eq:gammaballs}.

  The quantity we wish to bound may be given by
  \begin{align*}
      \sum_{y\not\in V}Q_f(0,y) 
      &= \sum_{y\not \in V}\sum_{c,d\in \F_2^n}P_f(c,d)P_f(c,y+d)\\
      &=\frac{1}{4^n}\sum_{c,d\in \F_2^n}|\widehat{\Delta_cf}(d)|^2\sum_{y\not \in V}\widehat{\Delta_cf}(y+d)|^2\\
      &=2\sum_{c\in \F_2^n}\Big(\sum_{y\not\in V}\frac{|\widehat{\Delta_cf}(y)|^2}{2^n}\Big)
      \Big(\sum_{z\in V}\frac{|\widehat{\Delta_cf}(z)|^2}{2^n}\Big).
  \end{align*}
  Keeping only the terms $c\in \{0,u\}$, we get that this is bounded from below by
  \begin{equation}\label{eq:foursums}
      2\Big(\sum_{y\not\in V}\frac{|\widehat{\Delta_0f}(y)|^2}{2^n}\Big)
      \Big(\sum_{z\in V}\frac{|\widehat{\Delta_0f}(z)|^2}{2^n}\Big)
      +
      2\Big(\sum_{y\not\in V}\frac{|\widehat{\Delta_uf}(y)|^2}{2^n}\Big)
      \Big(\sum_{z\in V}\frac{|\widehat{\Delta_uf}(z)|^2}{2^n}\Big).
  \end{equation}
  Expanding the definition of the Fourier transforms of the multiplicative derivatives gives that the above four sums  are bounded as follows
  \begin{align*}
     \sum_{y\not\in V}\frac{|\widehat{\Delta_0f}(y)|^2}{2^n}
     &=
     \frac{1}{2^{n+2}}\Exp_{x\in \F_2^n}\big(|f(x)|^2 - |f(x+u)|^2\big)^2\\
     &\geq \frac{1}{2^{2n+1}}\big(|f(0)|^2 - |f(u)|^2\big)^2.\\
     \sum_{z\in V}\frac{|\widehat{\Delta_0f}(z)|^2}{2^n}
     &=
     \frac{1}{2^{n+2}}\Exp_{x\in \F_2^n}\big(|f(x)|^2 + |f(x+u)|^2\big)^2\\
     &\geq \frac{1}{2^{2n+1}}\big(|f(0)|^2 + |f(u)|^2\big)^2.
  \end{align*}

    \begin{align*}
     \sum_{y\not\in V}\frac{|\widehat{\Delta_uf}(y)|^2}{2^n}
     &=
     \frac{1}{2^{n+1}}\Exp_{x\in \F_2^n}\big(|f(x)|^2|f(x+u)|^2 - \overline{f(x)}^2f(x+u)^2\big)\\
     &\geq \frac{1}{2^{2n}}\big(|f(0)|^2|f(u)|^2 - \Re\big(\overline{f(0)}^2f(u)^2\big)\big).\\
     \sum_{z\in V}\frac{|\widehat{\Delta_uf}(z)|^2}{2^n}
     &=
     \frac{1}{2^{n+1}}\Exp_{x\in \F_2^n}\big(|f(x)|^2|f(x+u)|^2 + \overline{f(x)}^2f(x+u)^2\big)\\
     &\geq \frac{1}{2^{2n}}\big(|f(0)|^2|f(u)|^2 + \Re\big(\overline{f(0)}^2f(u)^2\big)\big).
  \end{align*}

Combining these bounds gives that~\eqref{eq:foursums} is bounded from below by
\begin{equation}\label{eq:circles}
    \frac{1}{2^{4n+1}}|f(0)|^8(1-|y|)^2(1+|y|)^2
    +
    \frac{1}{2^{4n-1}}|f(0)|^8\big(|y| - \Re(y)\big)^2\big(|y| + \Re(y)\big)^2.
\end{equation}
Note that~$|f(0)|^8/2^{4n} = |\langle f,\phi\rangle|^8$.
We bound~\eqref{eq:circles} from below by using that the forbidden region of~$y$ in the complex plane given by~\eqref{eq:gammaballs} contains two segments of a narrow annulus around the complex unit circle (see Figure~\ref{fig:circles}).

\begin{center}
    \begin{figure}[ht]
        \includegraphics[width=6cm]{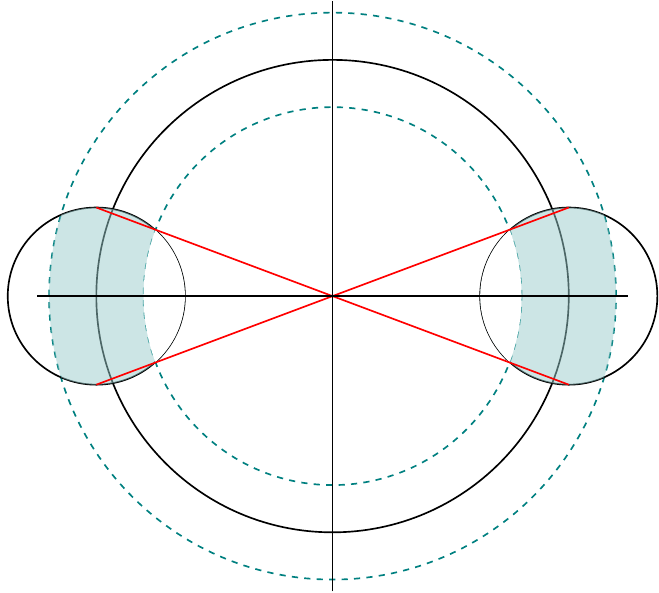}
    \caption{Forbidden regions for~$y$.}\label{fig:circles}
    \end{figure}
\end{center}

Choose the angles between the straight lines and the horizontal axis to be such that the distance from the origin to the small circles equals $r = \sqrt{1 - (\gamma-1/2)^2}$.

If~$y$ lies outside of the annulus, then the first term of~\eqref{eq:circles} is at least $\frac{1}{4}(\gamma - 1/2)^2|\langle f,\phi\rangle|^8$.
If~$y$ lies inside the annulus but outside of the small circles, then elementary trigonometry shows that the second term of~\eqref{eq:circles} is at least $\frac{1}{4}(\gamma - 1/2)^2|\langle f,\phi\rangle|^8$.

\end{proof}

\begin{lemma}\label{lem:Qfsmooth}
    Let $f: \F_2^n\to \C$.
    For $\gamma \in (\tfrac{1}{2}, 1)$, let~$\phi$ be a $\gamma$-approximate local maximizer for~$f$ such that $|\langle f,\phi\rangle| \geq\tau$.
    Then, for every proper subspace $T\subsetneq \mathcal L(\phi)$, we have
    \begin{equation*}
        Q_f\big(\mathcal L(\phi)\setminus T\big) \geq \tfrac{1}{4} \big(\gamma - \tfrac{1}{2}\big)^2 \tau^8.
    \end{equation*}
\end{lemma}

\begin{proof}
Using the properties of the Weyl operators and~\cite[Lemma~5.1]{Grewal2023}, it follows that there exists a unitary operator~$U\in\C^{2^n\times 2^n}$ and an invertible linear map~$S:\F_2^{2n}\to\F_2^{2n}$ such that $U\phi = 2^{n/2}\one_{\{0\}}$, $S(\mathcal L(\phi)) = \{0^n\}\times \F_2^n$, $S(T) = \{0^{n+1}\}\times\F_2^{n-1}$ and $2^{n/2}\one_{\{0\}}$ is a $\gamma$-approximate local maximizer of correlation for~$Uf$.
Together, these properties imply that
\begin{equation*}
    Q_f\big(\mathcal L(\phi)\setminus T\big)
    =
    Q_{Uf}\big(\{0^n\}\times (\F_2^n\setminus \{0\} \times \F_2^{n-1}) \big).
\end{equation*}
The result now follows from Lemma~\ref{lem:terrible}.
\end{proof}

\begin{proof}[ of Lemma~\ref{lem:nonrobust}]
If~$f$ does not $\eps$-robustly generate~$\mathcal L(\phi)$, then there is an $\eps$-approximate spectral set~$F$ such that $\mathcal L(\phi)\cap \vspan(F)$ is a proper subspace of~$\mathcal L(\phi)$.
It then follows from Lemma~\ref{lem:Qfsmooth} that
\begin{align*}
    Q_f\big(\mathcal L(\phi)\setminus \spec(f)\big) 
    &\geq Q_f(\mathcal L(\phi)\setminus \vspan(F)) - Q_f\big(\spec(f)\setminus \vspan(F)\big)\\
    &\geq \tfrac{1}{4} \big(\gamma - \tfrac{1}{2}\big)^2\tau^8 - \eps.
\end{align*}  
This proves the lemma.
\end{proof}

\subsubsection{Sampling the desired Lagrangian}

Putting the above ideas together gives the following result.

\begin{algorithm}[Lagrangian sampling]
For a 1-bounded function $f:\F_2^n\to\C$, $\gamma>1/2$ and $\tau>0$, define the algorithm $\textsc{LagrangianSample}(f,\gamma,\tau)$ as follows:
Let $t = \lceil\log_{1.08}(1/\tau)\rceil$.
Let $s$ be a uniformly random element from $\{0,1,\dots,t\}$.
Let~$f_0 = f$.
For each $i\in [t]$, generate 1-bounded functions $f_i:\F_2^n\to \C$ and vectors $(a_i,b_i)\in \F_2^n\times\F_2^n$ as follows:
\begin{itemize}
\item Let $(a_i,b_i)$ be a random sample from~$Q_{f_{i-1}}$.
\item For a uniformly distributed random sign~$\sigma_i$, let $f_i = \Pi_{a_i,b_i}^{\sigma_i}f_{i-1}$.
\end{itemize}
Return the basis obtained by the algorithm from Lemma~\ref{lem:rubustgen} on input~$f_s$ with parameters $\eps = (\gamma  -\frac{1}{2})^2\tau^8/8$ and $\delta = 1/2$.
\end{algorithm}

\begin{theorem}[Lagrangian sampling]\label{thm:lagrangian}
Let~$f:\F_2^n\to\C$ be a 1-bounded function.
Let~$\phi$ be a stabilizer state that is a $\gamma$-approximate local maximizer for~$f$ and satisfies $|\langle f,\phi\rangle|\geq \tau$.
Then, the algorithm $\textsc{LangrangeSample}(f,\gamma,\tau)$ returns a basis for a subspace $L\subseteq \F_2^n\times \F_2^n$ such that, with probability at least $((\gamma - 1/2)\tau)^{O(\log(1/\tau))}$, we have~$L=\mathcal L(\phi)$.
\end{theorem}

\begin{proof}
Given functions $g,g':\F_2^n\to \C$, define the following conditions:
\begin{itemize}
    \item Base condition $\mathrm{BC}(g)$: $\|g\|_\infty\leq 1$, $\phi$ is a $\gamma$-approximate local maximizer of correlation for~$g$ and $|\langle g,\phi\rangle|^2\geq \tau$.
    \item Robust generation $\mathrm{RG}(g)$: $\mathrm{BC}(g)$ holds and $g$ $\eps$-robustly generates~$\mathcal L(\phi)$.
    \item Energy increment $\mathrm{EI}(g,g')$: $\frac{|\langle g',\phi\rangle|^2}{\|g'\|_2^2} \geq 1.08\frac{|\langle g,\phi\rangle|^2}{\|g\|_2^2}$ and $\mathrm{BC}(g), \mathrm{BC}(g')$ hold.
\end{itemize}

For each $i\in\{0,1,\dots,t-1\}$ consider the success event 
\begin{equation*}
    \mathrm{succ}_i = \Big(\bigwedge_{j=0}^{i}\mathrm{EI}(f_j,f_{j+1})\Big)
    \vee
    \bigvee_{j=0}^i\mathrm{RG}(f_j).
\end{equation*}
Because the energy is capped by~1, we have that $\mathrm{succ}_t$ is union of events that one of the~$f_i$ $\eps$-robustly generates~$\mathcal L(\phi)$.

By Lemma~\ref{lem:boost} and Lemma~\ref{lem:nonrobust}, we have that 
\begin{equation}\label{eq:succ}
    \Pr\big[\mathrm{succ}_{i+1}\mid\mathrm{succ}_i\big]
    \geq\Pr\big[\mathrm{EI}(f_{i+1},f_{i+2})\vee \mathrm{RG}(f_{i+1})\mid \mathrm{BC}(f_{i+1})\big]\geq \frac{\eps}{2}.
\end{equation}
It follows from~\eqref{eq:succ} that
\begin{align*}
    \Pr\big[\mathrm{succ}_t\big]
    &=
    \Pr[\mathrm{succ}_0]\prod_{i=0}^{t-1}\Pr\big[\mathrm{succ}_{i+1}\mid\mathrm{succ}_i\big] \geq \Big(\frac{\eps}{2}\Big)^{t+1}.
\end{align*}

Conditioned on the event~$\mathrm{succ}_t$, we have that with probability~$\Omega(1/t)$ the function~$f_s$ $\eps$-robustly generates~$\mathcal L(\phi)$.
In that event, the algorithm returns~$\mathcal L(\phi)$ with probability at least~$1/2$.
\end{proof}

\subsection{From a good Lagrangian to a good stabilizer state} \label{sec:Stabilizer}

Suppose we know a basis for~$\Lcal(\phi)$, where~$\phi$ is a fixed (but unknown) $\gamma$-approximate local maximizer of correlation for~$f$ satisfying $|\langle f, \phi\rangle| \geq \tau$.
We now wish to learn~$\phi$ with good probability.

We show the following:

\begin{lemma}[Stabilizer sampling]\label{lem:linear}
    Let $f: \F_2^n \to \C$ be a $1$-bounded function and let~$\phi$ be a stabilizer state with $|\langle f, \phi\rangle| \geq \tau$.
    There is a randomized algorithm which, when given a basis $\{v_1, \dots, v_n\}$ for~$\Lcal(\phi)$, returns a random stabilizer state~$\psi$ with
    $$\Pr[\psi=\phi] \geq \tau^6/8.$$
    This algorithm makes $n \log n\, \poly(1/\tau)$ queries to~$f$ and runs in~$O(n^3)$ time.
\end{lemma}

\begin{proof}
Since~$\Lcal(\phi)$ is a Lagrangian subspace, we can write
\begin{equation} \label{eq:Lagrangian}
    \Lcal(\phi) = \big\{(h,\, Mh+w):\: h\in V,\, w\in V^{\perp}\big\}
\end{equation}
for some subspace $V\leq \F_2^n$ and symmetric matrix $M\in \F_2^{n\times n}$.
Moreover, we know that
\begin{equation} \label{eq:phi}
    \phi(x) = 2^{(n-\dim(V))/2} \one_{u+V}(x) (-1)^{x^{\mathsf T} Qx + c\cdot x} i^{|d\circ x|},
\end{equation}
where~$Q$ is the upper-diagonal part of matrix~$M$, $d$ is the diagonal of~$M$, and $c, u\in \F_2^n$ are vectors.

From the given basis $\{v_1, \dots, v_n\}$ of~$\Lcal(\phi)$ we can obtain, in $O(n^3)$ time, a basis for the subspace~$V$ and a matrix~$M$ such that identity~\eqref{eq:Lagrangian} holds.
In order to completely determine~$\phi$ as in equation~\eqref{eq:phi}, it only remains to find the correct coset~$u+V$ on which it is supported and its linear part~$(-1)^{c\cdot x}$.

Since~$f$ is bounded, the codimension of~$V$ is also bounded;
indeed,
$$\tau \leq |\langle f, \phi\rangle| \leq 2^{(n-\dim(V))/2} \E_{x\in \F_2^n} |f(x)| \one_{u+V}(x) \leq 2^{-(n-\dim(V))/2},$$
which implies that $n-\dim(V) \leq 2\log(1/\tau)$.
There are thus at most $2^{n-\dim(V)} \leq 1/\tau^2$ cosets~$w+V$ of~$V$ on which~$\phi$ can be supported.
Choosing a uniformly random vector~$w\in \F_2^n$, with probability at least~$\tau^2$ we obtain the correct coset $w+V = u+V$.

Now suppose we have found the correct coset $w+V$, and consider the function~$g$ given by
$$g(x) = \one_{w+V}(x) f(x)(-1)^{x^{\mathsf T} Qx} i^{-|d\circ x|}.$$
Letting~$c\in \F_2^n$ be the vector given in equation~\eqref{eq:phi} above, we have that
$$|\widehat{g}(c)| = \big|\E_{x\in \F_2^n} f(x) \one_{w+V}(x) (-1)^{x^{\mathsf T} Qx} i^{-|d\circ x|} (-1)^{c\cdot x}\big| = 2^{-(n-\dim(V))/2} |\langle f,\, \phi\rangle| \geq \tau^2.$$
Applying the Goldreich-Levin algorithm (Theorem~\ref{thm:GL}) to the function~$g$ with $\delta = 1/2$ and~$\tau$ substituted by~$\tau^2$, we obtain a list~$B\subseteq \F_2^n$ of size at most~$4/\tau^4$ which, with probability at least~$1/2$, satisfies
$$\big\{b\in \F_2^n:\: |\widehat{g}(b)| \geq \tau^2\big\} \subseteq B \subseteq \big\{b\in \F_2^n:\: |\widehat{g}(b)| \geq \tau^2/2\big\}.$$
Taking an element~$b\in B$ uniformly at random, we then get~$b=c$ with probability at least~$\tau^4/8$.

In conclusion, the (random) stabilizer state
$$\psi(x) := 2^{(n-\dim(V))/2} \one_{w+V}(x) (-1)^{x^{\mathsf T} Qx + b\cdot x} i^{|d\circ x|}$$
thus obtained will be equal to~$\phi$ with probability at least~$\tau^6/8$.
\end{proof}

\subsection{Approximate sampling from the convoluted distribution} \label{sec:Convoluted}

We now need to obtain an algorithmic procedure for sampling from the convoluted distribution~$Q_f$.
Given that $Q_f = P_f * P_f$, this would be easily done if we could sample from the simpler distribution~$P_f$.
However, doing so presents some difficulties:
by Parseval's identity we have
$$\sum_{b\in \F_2^n} P_f(a, b) = \sum_{b\in \F_2^n} \frac{|\widehat{\Delta_a f}(b)|^2}{2^n \|f\|_2^4} = \frac{\|\Delta_a f\|_2^2}{2^n \|f\|_2^4},$$
which can significantly vary with $a\in \F_2^n$.
As such, even if we can (approximately) sample from the marginal distribution $P_f(a, \cdot)/(\sum_{b} P_f(a, b))$ for a given $a\in \F_2^n$, there seems to be no easy way to sample~$a$ from a distribution proportional to $\|\Delta_a f\|_2^2$ using few queries to~$f$.

Our solution is to ignore this difficulty and instead sample $a\in \F_2^n$ uniformly at random, followed by sampling~$b$ with probability close to $|\widehat{\Delta_a f}(b)|^2$.
We thereby obtain a sample $(a, b)$ from some probability distribution~$\nu_f$ that approximates the \emph{non-probability measure} $\|f\|_2^4 P_f$ in a \emph{fairly weak sense}.
Upon convolving~$\nu_f$ with itself, this distribution gets smoothened out and we manage to obtain the following result:

\begin{theorem}[Convoluted sampling] \label{thm:sampling}
    Let $f: \F_2^n \to \C$ be a $1$-bounded function.
    There is a randomized sampling procedure that makes $n\log n\, \poly(1/\xi)$ queries to~$f$ and, with probability at least $1-1/n^2$, samples from a probability distribution~$\mu_f$ that satisfies
    \beqn
    \big|\mu_f(F) - \|f\|_2^8 Q_f(F)\big| \leq \frac{\xi |F|}{2^n} \quad \text{for all $F\subseteq \F_2^{2n}$.}
    \eeqn
\end{theorem}

Note that, unless $\|f\|_2 = 1$, the expression $\|f\|_2^8 Q_f$ is \emph{not} a probability measure.
It would then be impossible for our samplable distribution~$\mu_f$ to approximate this measure in a more obvious way such as total variation distance.
However, since all the events that are important for our algorithm correspond to isotropic sets (and thus have size at most~$2^n$), the approximation given in Theorem~\ref{thm:sampling} is essentially just as good as total variation distance for our purposes.

\begin{proof}[ of Theorem~\ref{thm:sampling}]
Without loss of generality, we may assume that $\xi \leq 1/2$ and that~$1/\xi$ is an integer, so we do not need to deal with floor functions.
Given $a\in \F_2^n$, we can use the Goldreich-Levin algorithm (Theorem~\ref{thm:GL}) on $\Delta_a f$ to find a set $B_a \subseteq \F_2^n$ of size at most $64/\xi^2$ which, with probability at least $1-\eta$, satisfies
\begin{equation*}
    \big\{b\in \F_2^n:\: |\widehat{\Delta_a f}(b)| \geq \xi/4\big\} \subseteq B_a \subseteq \big\{b\in \F_2^n:\: |\widehat{\Delta_a f}(b)| \geq \xi/8\big\}.
\end{equation*}
This takes $n \log n\, \poly(\xi^{-1} \log(\eta^{-1}))$ queries to~$f$.

Next, we query~$f$ a further $\poly(\xi^{-1} \log(\eta^{-1}))$-many times to obtain nonnegative numbers $\{\lambda_a(b): b\in B_a\}$ such that, with probability at least $1-\eta$, we have
$$\big||\widehat{\Delta_a f}(b)|^2 - \lambda_a(b)\big| \leq \xi^4 \quad \text{for all $b\in B_a$}$$
(see Lemma~\ref{lem:fourest}).
Then, with probability at least $1-\eta$, we have
\begin{equation*}
    \sum_{b\in B_a} \lambda_a(b) \leq \sum_{b\in B_a} \big(|\widehat{\Delta_a f}(b)|^2 + \xi^4\big) \leq \|\Delta_a f\|_2^2 + \xi^4 |B_a| \leq 1 + 4\xi^2.
\end{equation*}
If $\sum_{b\in B_a} \lambda_a(b) > 1+4\xi^2$ (which happens with probability at most~$\eta$), replace the $\lambda_a(b)$ by zero.

Now we increase~$B_a$ arbitrarily to a superset $B_a' \subseteq \F_2^n$ of size $|B_a| + 4/\xi$, and define the function $\nu_a: \F_2^n \to [0, 1]$ by
$$\nu_a(b) = \frac{\lambda_a(b)}{1+4\xi^2}\, \text{ if $b\in B_a$}, \quad \nu_a(b) = \frac{\xi}{4}\Big(1- \frac{1}{1+4\xi^2} \sum_{b\in B_a} \lambda_a(b)\Big)\, \text{ if $b\in B_a' \setminus B_a$}$$
and $\nu_a(b) = 0$ if $b\notin B_a'$.
It is clear that~$\nu_a$ is a probability measure with $|\supp(\nu_a)| \leq |B_a'| \leq 68/\xi^2$ and, with probability at least $1-2\eta$, it satisfies
$$\big|\nu_a(b) - |\widehat{\Delta_a f}(b)|^2\big| \leq \frac{\xi}{4} \quad \text{for all $b\in \F_2^n$.}$$

Define the probability distribution~$\nu_f$ on~$\F_2^{2n}$ by $\nu_f(a, b) = \nu_a(b)/2^n$.
This distribution is easy to sample from:
sample $a\in \F_2^n$ uniformly at random, then compute~$\nu_a$ on $\supp(\nu_a)$ using $n \log n\, \poly(\xi^{-1} \log(\eta^{-1}))$ queries to~$f$, then sample $b\in \supp(\nu_a)$ according to~$\nu_a$.

Denote
\begin{equation*}
    A = \big\{a\in \F_2^n:\: \big|\nu_a(b) - |\widehat{\Delta_a f}(b)|^2\big| > \xi/4\, \text{ for some $b\in \F_2^n$}\big\}.  
\end{equation*}

Since $\Pr[a\in A] \leq 2\eta$ independently for all $a\in \F_2^n$, we conclude from Chernoff's bound that $\Pr\big[|A| \geq 4\eta \cdot 2^n\big] \leq 1 - 1/n^2$.
Moreover, by boundedness of~$f$ and~$\nu_a$, we have
$$\big|\nu_a(b) - |\widehat{\Delta_a f}(b)|^2\big| \leq \frac{\xi}{4} + \one_A(a) \quad \text{for all $a, b\in \F_2^n$.}$$

Now let~$F\subseteq \F_2^{2n}$ be any set.
Writing $\tilde{P}_f(a, b) := \|f\|_2^4 P_f(a, b) = 2^{-n} |\widehat{\Delta_a f}(b)|^2$, we have that
\begin{align*}
    \big|\tilde{P}_f * (\tilde{P}_f - \nu_f) (F)\big|
    &= \bigg| \sum_{c, d\in \F_2^n} \tilde{P}_f(c,d) \sum_{(a,b) \in F} \big(\tilde{P}_f(a+c, b+d) - \nu_f(a+c, b+d)\big) \bigg| \\
    &\leq\sum_{c, d\in \F_2^n} \tilde{P}_f(c,d) \sum_{(a,b) \in F} \frac{\big||\widehat{\Delta_{a+c}f}(b+d)|^2 - \nu_{a+c}(b+d)\big|}{2^n} \\
    &\leq\sum_{c, d\in \F_2^n} \tilde{P}_f(c,d) \sum_{(a,b) \in F} \frac{\xi/4 + \one_A(a+c)}{2^n} \\
    &\leq \frac{\xi}{4} \frac{|F|}{2^n} + \frac{1}{2^n} \sum_{(a,b) \in F} \sum_{c, d\in \F_2^n} \tilde{P}_f(c,d) \one_A(a+c) \\
    &= \frac{\xi}{4} \frac{|F|}{2^n} + \frac{1}{2^n} \sum_{(a,b) \in F} \sum_{c\in a+A} \sum_{d\in \F_2^n} \tilde{P}_f(c,d).
\end{align*}
Noting that
$$\sum_{d\in \F_2^n} \tilde{P}_f(c,d) = \frac{1}{2^n} \sum_{d\in \F_2^n} |\widehat{\Delta_c f}(d)|^2 = \frac{1}{2^n} \|\Delta_c f\|_2^2 \leq \frac{1}{2^n},$$
we conclude that
$$\big|\tilde{P}_f * (\tilde{P}_f - \nu_f) (F)\big| \leq \frac{\xi}{4} \frac{|F|}{2^n} + \frac{|F|}{2^n} \frac{|A|}{2^n}.$$
Similarly we obtain
$$\big|\nu_f * (\tilde{P}_f - \nu_f)(F)\big| \leq \frac{\xi}{4} \frac{|F|}{2^n} + \frac{|F|}{2^n} \frac{|A|}{2^n},$$
and thus
$$\big|\nu_f * \nu_f(F) - \tilde{P}_f * \tilde{P}_f(F)\big| \leq \frac{\xi}{2} \frac{|F|}{2^n} + 2\frac{|F|}{2^n} \frac{|A|}{2^n}.$$

Taking $\eta = \xi/16$ and denoting $\mu_f = \nu_f * \nu_f$, we conclude that, with probability at least $1-1/n^2$, we have
\beqn
    \big|\mu_f(F) - \|f\|_2^8 Q_f(F)\big| \leq \frac{\xi |F|}{2^n} \quad \text{for all $F\subseteq \F_2^{2n}$.}
\eeqn
Note that we can sample from~$\mu_f$ by sampling independent pairs $(a,b)$, $(c,d)$ according to~$\nu_f$ and returning $(a+c,\, b+d)$.
The result follows.
\end{proof}

\subsection{Finding all good stabilizer states} \label{sec:List}

Now we combine everything we have done into a single algorithm that, with high probability, outputs a bounded-size list containing all $\gamma$-approximate local maximizers of correlation~$\phi$ with~$f$ satisfying $|\langle f, \phi\rangle| \geq \tau$.

Let~$\mu_f$ be the random probability distribution from Theorem~\ref{thm:sampling} and suppose that it satisfies the conclusion of the theorem.

\subsubsection{Robust generation}
We approximately implement the algorithm from Lemma~\ref{lem:rubustgen} by substituting samples from~$Q_f$ by samples from~$\mu_f$. 
The number of samples we use now depends on the value $p = \mu_f\big(\Spec(f)\big)$.
By the relationship between~$\mu_f$ and~$Q_f$ and the fact that $\|f\|_2 \geq \tau$, an analysis similar to the proof of Lemma~\ref{lem:rubustgen} shows that with a factor of~$O(1/\tau^8)$ more samples from~$\mu_f$ we obtain a basis for a subspace of~$L\subseteq\F_2^n\times\F_2^n$ such that with high probability~$L = \mathcal L(\phi)$, provided $\xi \leq \eps\tau^8/2$.

Each sample from~$\mu_f$ then costs $n\log n\poly(1/\xi)$ queries to~$f$.
Hence, the total query complexity of this algorithm is $n^2\log n\,\poly(1/\xi)$.

\subsubsection{Non-robust generation}

Then, if~$f$ not does $\eps$-robustly generates~$\mathcal L(\phi)$, we have from Lemma~\ref{lem:nonrobust} that
\begin{align*}
    \mu_f\big(\mathcal L(\phi)\setminus\Spec(f)\big) &\geq \tau^8Q_f\big(\mathcal L(\phi)\setminus\Spec(f)\big) - \xi\\
    &\geq \frac{1}{8}\big(\gamma - \tfrac{1}{2}\big)^2\tau^{16},
\end{align*}
provided $\xi \leq \frac{1}{8}\big(\gamma - \tfrac{1}{2}\big)^2\tau^{16}$.

We approximately implement the algorithm $\textsc{LagrangeSample}(f,\tau)$ by substituting its samples from~$Q_{f_i}$ with samples from $\mu_{f_i}$.
Using an analysis similar to the proof of Theorem~\ref{thm:lagrangian}, now using $\eps = \frac{1}{8}\big(\gamma - \tfrac{1}{2}\big)^2\tau^{16}$, we get a basis for a subspace~$L$ that satisfies~$L = \mathcal L(\phi)$ with probability at least $((\gamma - 1/2)\tau)^{O(\log(1/\tau))}$.

Note that we can query each projected function~$f_i$ using~$2^i$ queries to~$f$.
A sample from $\mu_{f_i}$ therefore costs at most~$n\log n\,\poly(1/\tau, 1/(\gamma - 1/2))$ queries to~$f$.
So the generation of~$f_1,\dots,f_t$ a query complexity of the same order.
The run of the approximate robust generation algorithm described in the previous section then has query complexity $n^2\log n\,\poly(1/\tau, 1/(\gamma - 1/2))$.

\subsubsection{List-decoding stabilizer states}

Combining the algorithm from the previous section with Lemma~\ref{lem:linear} then gives an algorithm that does the following:
For any fixed $\gamma$-approximate local maximizer of correlation~$\phi$ for~$f$ with $|\langle f,\phi\rangle| \geq \tau$, makes $n^2\log n\,\poly(1/\tau, 1/(\gamma-1/2))$ queries to~$f$ and returns~$\phi$ with probability at least~$p=((\gamma - 1/2)\tau)^{O(\log(1/\tau))}$.

Since this is for a fixed $\gamma$-approximate local maximizer of correlation~$\phi$, repeating this algorithm~$O((1/p)\log(1/p))$ times gives the completely list of all such stabilizer states with good probability.
This concludes the proof of Theorem~\ref{thm:list}.

\section{Proving the main results}

We now use the list-decoding algorithm given in Theorem~\ref{thm:list} to construct our quadratic Goldreich-Levin algorithm.

\begin{proof}[ of Theorem~\ref{thm:main}]
The main idea here is to apply the algorithm from Theorem~\ref{thm:list} with suitably chosen parameters to obtain a bounded-size list containing all ``good'' stabilizer states, and then replace each of these good stabilizer states by a bounded number of quadratic phase functions.
Each such quadratic phase~$(-1)^q$ is obtained from its associated stabilizer state~$\phi$ by extending its support from (a coset of) a subspace~$V$ the whole domain~$\F_2^n$.
We end the proof by showing that, with high probability, one of the quadratic phases thus obtained has almost-maximal correlation with~$f$;
by querying~$f$ a bounded number of times, we can estimate all of these correlations and pick up the highest one.

The full algorithm is given as follows:
\begin{enumerate}
    \item Apply the algorithm from Theorem~\ref{thm:list} with parameters $\tau = \eps$ and $\gamma = 1/2 + \eps^2$.
    We obtain a list~$L$ of size $\log(1/\delta) (1/\eps)^{O(\log(1/\eps))}$ which, with probability at least $1-\delta$, contains all stabilizer states that are $(1/2+\eps^2)$-approximate local maximizers of correlation for~$f$ and have correlation at least~$\eps$ with~$f$.
    
    \item Remove from~$L$ every stabilizer state whose support has codimension larger than $2\log(1/\eps)$.
    If~$L$ becomes empty after this step, end the algorithm and return the constant function $p \equiv 0$.
    Otherwise, initialize a list~$L'$ to be empty and continue.
    
    \item For each stabilizer state $\phi\in L$, do the following:

    \noindent Write $\phi(x) = 2^{(n-d)/2} \one_{u+V}(x) (-1)^{q(x)} i^{|c\circ x|}$, where~$V$ is a subspace of dimension~$d$, $q: \F_2^n \to \F_2$ is a quadratic function and $u, c \in \F_2^n$ are vectors.
    Let~$U = \{c\}^\perp$ and let~$v\in \F_2^n$ satisfy~$c\cdot v = 1$, so that any $x\in\F_2^n$ has a representation of the form $x = y+bv$ for some $y\in U$ and $b\in \F_2$.
    Using polynomial interpolation, find the polynomial $r\in \F_2[x_1,\dots,x_n]$ of degree at most~2 such that $(-1)^{r(y+bv)} = i^{|c\circ y|  - 2|bc\circ y\circ v|}$.
    Add to~$L'$ the quadratic functions $x\mapsto r(x)+q(x)+y\cdot x$ and $x\mapsto r(x)+q(x)+ (y+c)\cdot x$, for every $y\in V^{\perp}$.
    
    \item Query~$f$ at $m = \poly(1/\eps,\, \log 1/\delta)$ randomly chosen points $x_1, \dots, x_m \in \F_2^n$ and compute
    $$\Est_q := \frac{1}{m} \sum_{j=1}^m f(x_j) (-1)^{q(x_j)}$$
    for all quadratic functions~$q$ in~$L'$.
    Output the one that attains the maximum value of~$|\Est_q|$.
\end{enumerate}

Note that, for each $\phi\in L$, the number of quadratic functions we add to~$L'$ at step~$(3)$ is at most~$2^{n-d+1}$.
Since $n-d \leq 2\log(1/\eps)$ because of step~$(2)$, it follows that the final list~$L'$ has size at most
$$2^{n-d+1} |L| \leq 2|L|/\eps^2 = \log(1/\delta) (1/\eps)^{O(\log(1/\eps))}.$$
The query and time complexities of the algorithm above thus match those stated in Theorem~\ref{thm:main}.

Denote the (random) quadratic function output by this algorithm by~$p$.
We will show that, with probability at least~$1-2\delta$, this function satisfies
\begin{equation} \label{eq:high_corr}
    |\langle f,\, (-1)^{p(\cdot)}\rangle| > \|f\|_{u^3} - \eps;
\end{equation}
this will complete the proof of the theorem.

We may focus on the case where $\|f\|_{u^3} \geq \eps$, as otherwise any quadratic function will satisfy~$\eqref{eq:high_corr}$.
We can also assume that~$\eps \leq 1/100$, which will allow us to bound certain expressions more easily.
The heart of the argument is given in the following result:

\begin{lemma} \label{lem:find_quad}
    Assume that~$\eps \leq 1/100$ and~$\|f\|_{u^3} \geq \eps$.
    Then, with probability at least~$1-\delta$, there exists a quadratic function~$q$ in~$L'$ satisfying
    $$|\langle f,\, (-1)^{q(\cdot)}\rangle| \geq \|f\|_{u^3} - \eps/2.$$
\end{lemma}

\begin{proof}
Let $p^*: \F_2^n \to \F_2$ be a quadratic function attaining maximum correlation with~$f$:
$$\big|\E_{x\in \F_2^n} f(x) (-1)^{p^*(x)}\big| = \|f\|_{u^3}.$$
Consider the stabilizer state $\phi_0 := (-1)^{p^*(\cdot)}$, and denote $\gamma = 1/2+\eps^2$.
If~$\phi_0$ is a $\gamma$-approximate local maximizer of fidelity with~$f$, then with probability at least~$1-\delta$ it will appear in list~$L$
(and thus~$p^*$ will appear in~$L'$).

Now suppose~$\phi_0$ is not a $\gamma$-approximate local maximizer of correlation with~$f$.
There must then exist a ``neighbor'' stabilizer state~$\phi_1$ satisfying
$$|\langle\phi_0,\, \phi_1\rangle|^2 = 1/2 \quad \text{and} \quad |\langle f,\, \phi_1 \rangle|^2 > \gamma^{-1} |\langle f,\, \phi_0 \rangle|^2.$$
If~$\phi_1$ is a $\gamma$-approximate local maximizer, then it will appear in list~$L$ with probability at least~$1-\delta$.
Otherwise, we can keep choosing stabilizer states~$\phi_{i+1}$ satisfying
$$|\langle\phi_i,\, \phi_{i+1}\rangle|^2 = 1/2 \quad \text{and} \quad |\langle f,\, \phi_{i+1} \rangle|^2 > \gamma^{-1} |\langle f,\, \phi_i \rangle|^2$$
until at last we arrive at some~$\phi_t$ which is a $\gamma$-approximate local maximizer of correlation with~$f$.
This must stop at some point because $|\langle f,\, \phi_0 \rangle| = \|f\|_{u^3} \geq \eps$ and we always have $|\langle f,\, \phi_i \rangle| \leq \|f\|_2$ by Cauchy-Schwarz.
The final stabilizer state~$\phi_t$ will then appear in list~$L$ with probability at least~$1-\delta$, and it satisfies
\begin{equation} \label{eq:corr_phi_t}
    |\langle f,\, \phi_t \rangle|^2 > \gamma^{-t} |\langle f,\, \phi_0 \rangle|^2 = (1/2+\eps^2)^{-t} \|f\|_{u^3}^2.
\end{equation}

Let us write
$$\phi_t(x) = 2^{(n-d)/2} \one_{u+V}(x) (-1)^{q^*(x)} i^{|c\circ x|},$$
where~$V$ is a subspace of dimension $d$, $q^*: \F_2^n \to \F_2$ is a quadratic function and $u, c \in \F_2^n$ are vectors.
Up to replacing~$\phi_t$ by $i^{-|c\circ u|} \phi_t$, we may assume that either~$c=0$ or $c\notin V^{\perp}$.
Since
$$\eps \leq |\langle f,\, \phi_0 \rangle| \leq |\langle f,\, \phi_t \rangle| \leq 2^{(n-d)/2} \E_{x\in \F_2^n} \one_{u+V}(x) |f(x)| = 2^{-(n-d)/2},$$
we conclude that $n-d \leq 2\log(1/\eps)$, and so~$\phi_t$ will not get removed in step~$(2)$ of the algorithm.

Next, we relate the dimension~$d$ of~$V$ with the number~$t$ of steps we took until we arrived at~$\phi_t$.
For each $0\leq i\leq t$, denote by $\dim(\phi_i)$ the dimension of the subspace on which the $i$-th stabilizer state~$\phi_i$ is supported.
Since $|\langle\phi_i,\, \phi_{i+1}\rangle|^2 = 1/2$ while $|\phi_j(\cdot)| = 2^{(n-\dim(\phi_j))/2} \one_{\supp(\phi_j)}(\cdot)$, we conclude that
$\dim(\phi_{i+1}) \geq \dim(\phi_i)-1$.
Moreover, in the case where $\dim(\phi_{i+1}) = \dim(\phi_i)-1$, the two stabilizer states~$\phi_i$ and~$\phi_{i+1}$ must be proportional to one another inside the support of~$\phi_{i+1}$.
As~$\phi_0 = (-1)^{p^*(\cdot)}$ while~$\phi_t$ has a nontrivial non-classical component~$i^{|c\circ x|}$ if $c\notin V^{\perp}$, it follows that $\dim(\phi_t) \geq \dim(\phi_0) - t + \one_{c\notin V^{\perp}}$.
We conclude that $t \geq n-d + \one_{c\notin V^{\perp}}$.

Using the fact that $\E_{y\in V^{\perp}} (-1)^{y\cdot x} = \one_V(x)$, we see that
\begin{align*}
    |\langle f,\, \phi_t \rangle|^2
    &= 2^{n-d} \big|\E_{x\in \F_2^n} f(x) \one_{V}(x+u) (-1)^{q^*(x)} i^{-|c\circ x|}\big|^2 \\
    &= 2^{n-d} \big|\E_{y\in V^{\perp}} \E_{x\in \F_2^n} f(x) (-1)^{y\cdot (x+u)} (-1)^{q^*(x)} i^{-|c\circ x|}\big|^2 \\
    &\leq 2^{n-d} \E_{y\in V^{\perp}} \big|\E_{x\in \F_2^n} f(x) (-1)^{q^*(x) + y\cdot x} i^{-|c\circ x|}\big|^2,
\end{align*}
and thus there exists some $y^* \in V^{\perp}$ such that
\begin{equation} \label{eq:ystar}
    \big|\E_{x\in \F_2^n} f(x) (-1)^{q^*(x) + y^* \cdot x} i^{-|c\circ x|}\big|^2 \geq 2^{-(n-d)} |\langle f,\, \phi_t \rangle|^2.
\end{equation}
Recall that, if $\phi_t \in L$ (which happens with probability at least~$1-\delta$), then the quadratic functions $x\mapsto q^*(x) + y^*\cdot x$ and $x\mapsto q^*(x) + (y^*+c)\cdot x$ will both be in~$L'$.
It then suffices to show that one of these functions has correlation at least $\|f\|_{u^3}-\eps/2$ with~$f$.

We separate the proof into two cases:
$c=0$ or $c\notin V^{\perp}$.
If~$c=0$, then by~\eqref{eq:corr_phi_t} and~\eqref{eq:ystar} we have
$$\big|\E_{x\in \F_2^n} f(x) (-1)^{q^*(x) + y^* \cdot x}\big|^2 \geq 2^{-(n-d)} (1/2+\eps^2)^{-t} \|f\|_{u^3}^2.$$
Using that $n-d \leq 2\log(1/\eps)$ and $t\geq n-d$, we conclude that
$$2^{-(n-d)} (1/2+\eps^2)^{-t} \geq 2^{-(n-d)} (1/2+\eps^2)^{-(n-d)} \geq (1+2\eps^2)^{-2\log(1/\eps)}.$$
This last expression is larger than $(1-\eps/2)^2$ when $\eps \leq 1/100$, which implies that $\big|\E_{x\in \F_2^n} f(x) (-1)^{q^*(x) + y^* \cdot x}\big| \geq \|f\|_{u^3} - \eps/2$ as wished.

In the case where $c\notin V^{\perp}$, let~$U \subseteq \F_2^n$ be the subspace orthogonal to~$c$, and let~$v \in \F_2^n\setminus U$.
Then, for any function $g: \F_2^n \to \C$, we have
\begin{align}
    \big|\E_{x\in \F_2^n} g(x) i^{-|c\circ x|}\big|^2
    &= \big|\E_{b\in \F_2} \E_{y\in U} g(y+bv) i^{-|c\circ (y+bv)|} \big|^2.\label{eq:xgisfr}
\end{align}

We have that $|c\circ (y+bv)| = |c\circ y| + |bc\circ v| - 2|bc\circ y\circ v|$.
Define $h:\F_2^n\to \C$ by $h(y+bv) = i^{|c\circ y|  - 2|bc\circ y\circ v|}$.
Then~$h$ is a classical polynomial phase function of degree at most~2.
In other words, there exists a polynomial $r\in \F_2[x_1,\dots,x_n]$ of degree at most~2 such that~$h(x) = (-1)^{r(x)}$.
By the triangle inequality and the Cauchy-Schwarz inequality, we get that~\eqref{eq:xgisfr} is bounded from above by
\begin{equation*}
    \Exp_{b\in \F_2}\big|\Exp_{y\in U}g(y+bv)(-1)^{r(y+bv)}\big|^2
\end{equation*}

By Parseval's identity on~$\F_2$ we get
\begin{align*}
    \E_{b\in \F_2} |\E_{y\in U} g(y+bv)(-1)^{r(y+bv)}|^2
    &= \sum_{a\in \F_2} |\E_{b\in \F_2} \E_{y\in U} g(y+bv) (-1)^{r(y+bv) + ab}|^2 \\
    &= \sum_{a\in \F_2} |\E_{b\in \F_2} \E_{y\in U} g(y+bv) (-1)^{r(y+bv) + a c\cdot (y+bv)}|^2 \\
    &= |\E_{x\in \F_2^n} g(x)(-1)^{r(x)}|^2 + |\E_{x\in \F_2^n} g(x) (-1)^{r(x) + c\cdot x}|^2.
\end{align*}
We conclude that
$$|\E_{x\in \F_2^n} g(x)(-1)^{r(x)}|^2 + |\E_{x\in \F_2^n} g(x) (-1)^{r(x) + c\cdot x}|^2 \geq \big|\E_{x\in \F_2^n} g(x) i^{-|c\circ x|}\big|^2.$$
Using this inequality for the function $g(x) = f(x) (-1)^{q^*(x) + y^*\cdot x}$, we obtain
\begin{align*}
    \max\big\{|\E_{x\in \F_2^n} f(x) (-1)^{r(x) + q^*(x) + y^* \cdot x}|^2,\: &|\E_{x\in \F_2^n} f(x) (-1)^{r(x)+q^*(x) + (y^*+c) \cdot x}|^2 \big\} \\
    &\geq \frac{1}{2} \big|\E_{x\in \F_2^n} f(x) (-1)^{q^*(x) + y^* \cdot x} i^{-|c\circ x|}\big|^2 \\
    &\geq \frac{1}{2^{n-d+1}} |\langle f,\, \phi_t \rangle|^2 \\
    &\geq \frac{1}{2^{n-d+1}} \Big(\frac{1}{1/2+\eps^2}\Big)^t \|f\|_{u^3}^2,
\end{align*}
where we used inequalities~\eqref{eq:ystar} and~\eqref{eq:corr_phi_t} respectively.
Since in this case we have $t\geq n-d+1$ and $n-d\leq 2\log(1/\eps)$, the last expression is at least
$$\Big(\frac{1}{1 + 2\eps^2}\Big)^{2\log(1/\eps)+1} \|f\|_{u^3}^2 \geq (1-\eps/2)^2 \|f\|_{u^3}^2$$
(where we use that~$\eps \leq 1/100$).
This concludes the proof.
\end{proof}

Using our bound on the size of the list~$L'$ and the Chernoff bound we conclude that, with probability at least~$1-\delta$, we have
$$\big|\Est_q - \langle f,\, (-1)^{q(\cdot)}\rangle\big| < \eps/4 \quad \text{for all $q\in L'$.}$$
If this is the case, then
$$|\langle f,\, (-1)^{p(\cdot)}\rangle| > |\Est_p| - \frac{\eps}{4} = \max_{q \in L'} |\Est_q| - \frac{\eps}{4} > \max_{q \in L'} |\langle f,\, (-1)^{q(\cdot)}\rangle| - \frac{\eps}{2}.$$
By Lemma~\ref{lem:find_quad} we have that $\max_{q \in L'} |\langle f,\, (-1)^{q(\cdot)}\rangle| \geq \|f\|_{u^3} - \eps/2$ with probability at least~$1-\delta$.
This implies that inequality~\eqref{eq:high_corr} holds with probability at least~$1-2\delta$, as wished.
\end{proof}

The optimal self-correction of quadratic Reed-Muller codes easily follows from the quadratic Goldreich-Levin algorithm:

\begin{proof}[ of Corollary~\ref{cor:ReedMuller}]
Query access to a Boolean function $f: \F_2^n \to \F_2$ gives query access to the bounded function $g(x) := (-1)^{f(x)}$.
Note that, for any Boolean function $q: \F_2^n \to \F_2$, we have
\begin{equation} \label{eq:corr_to_dist}
    \E_{x\in \F_2^n} g(x) (-1)^{q(x)} = 1 - 2\Pr_{x\in \F_2^n}[f(x) \neq q(x)] = 1-2\dist(f,\, q).
\end{equation}
Applying Theorem~\ref{thm:main} to~$g$ (with~$\eps$ substituted by~$\eps/4$ and $\delta=1/6$), we obtain a quadratic polynomial $p: \F_2^n \to \F_2$ which, with probability at least~$5/6$, satisfies
$$\big|\Exp_{x\in \F_2^n} g(x) (-1)^{p(x)}\big| > \max_{q \text{ quadratic}} \big|\Exp_{x\in \F_2^n} g(x) (-1)^{q(x)}\big| - \eps/4.$$
Using~$O(1/\eps^2)$ further queries to~$g$, we can differentiate (with probability at least~$5/6$) between the two cases
$$\Exp_{x\in \F_2^n} g(x) (-1)^{p(x)} \geq \eps/8 \quad \text{and} \quad \Exp_{x\in \F_2^n} g(x) (-1)^{p(x)} < -\eps/8;$$
in the later case, we replace~$p$ by its negation~$\one+p$.
The guarantees stated in the corollary immediately follow from those of Theorem~\ref{thm:main} together with equation~\eqref{eq:corr_to_dist}.
\end{proof}

Next we use the recent resolution of Marton's conjecture~\cite{GowersGMT2025} to obtain a polynomial algorithmic inverse theorem for the Gowers $U^3$-norm.

\begin{proof}[ of Corollary~\ref{cor:PGI}]
By \cite[Corollary~1.6]{GowersGMT2025} there exists a constant $c>1$ such that, whenever a bounded function $f: \F_2^n \to \C$ satisfies $\|f\|_{U^3} \geq \gamma$, there is a quadratic polynomial $q: \F_2^n \to \F_2$ with
$\big|\E_{x\in \F_2^n} f(x) (-1)^{q(x)}\big| \geq (\gamma/2)^c$.
Apply Theorem~\ref{thm:main} to~$f$ with $\eps = \gamma^c/2^{c+1}$ and $\delta = 1/3$;
we obtain a quadratic polynomial~$p$ which, with probability at least~$2/3$, satisfies
$$\big|\E_{x\in \F_2^n} f(x) (-1)^{p(x)}\big| > \big|\E_{x\in \F_2^n} f(x) (-1)^{q(x)}\big| - \frac{\gamma^c}{2^{c+1}} \geq (\gamma/2)^{c+1}.$$
The result follows, with $C = c+1$.
\end{proof}

Finally, we can obtain our algorithmic decomposition theorem by combining the last corollary with the framework developed by Tulsiani and Wolf~\cite{TulsianiW2014}.

\begin{proof}[ of Corollary~\ref{cor:decomposition}]
Denote $B := 1/(2\eps)$.
Corollary~\ref{cor:PGI} provides an algorithm which, when given query access to a function $f: \F_2^n \to \{z\in \C:\: |z| \leq B\}$ satisfying $\|f\|_{U^3} \geq \eps$, outputs with probability~$1-\delta$ a quadratic function $p: \F_2^n \to \F_2$ such that $|\langle f,\, (-1)^{q}\rangle| \geq \eps^{2C}$.
This algorithm makes $n^2 \log n\, \poly(\eps^{-1} \log(\delta^{-1}))$ queries to~$f$ and takes~$O(n^3)$ time.

The result now follows by applying \cite[Theorem~3.1]{TulsianiW2014} to this algorithm and the norm~$\|\cdot\|_{U^3}$.
\end{proof}

We note that, by replacing our use of \cite[Theorem~3.1]{TulsianiW2014} by \cite[Theorem~3.3]{kim2023cubic}, it is possible to do away with the $L_1$-error function~$h$ in this decomposition at the price of increasing the number of quadratic phase functions to $\exp(\poly(1/\eps))$.
It is at present unclear whether there exists a decomposition that attains the best of both worlds, even if one is to ignore the algorithmic aspects.

\subsection*{Acknowledgments}

JB was supported by the Dutch Research Council (NWO) as part of the NETWORKS programme (Grant No.~024.002.003).
DCS was supported by the Engineering and Physical Sciences Research Council grant on Robust and Reliable Quantum Computing (RoaRQ), Investigation 005 [grant reference EP/W032635/1].


\begin{thebibliography}{BSRZTW14}

\bibitem[AD25]{arunachalam2025polynomial}
Srinivasan Arunachalam and Arkopal Dutt.
\newblock Polynomial-time tolerant testing stabilizer states.
\newblock In {\em Proceedings of the 57th Annual ACM Symposium on Theory of Computing (STOC)}, Prague, Czech Republic, 2025.
\newblock To appear. Available at \url{https://arxiv.org/abs/2408.06289}.

\bibitem[BLR93]{blum1993self}
Manuel Blum, Michael Luby, and Ronitt Rubinfeld.
\newblock Self-testing/correcting with applications to numerical problems.
\newblock {\em Journal of Computer and System Sciences}, 47(3):549--595, 1993.
\newblock \href {https://doi.org/10.1016/0022-0000(93)90044-W} {\path{doi:10.1016/0022-0000(93)90044-W}}.

\bibitem[BSRZTW14]{Ben-Sasson2014}
Eli Ben-Sasson, Noga Ron-Zewi, Madhur Tulsiani, and Julia Wolf.
\newblock Sampling-based proofs of almost-periodicity results and algorithmic applications.
\newblock In {\em Proceedings of the 41st International Colloquium on Automata, Languages, and Programming (ICALP)}, pages 955--966, 2014.
\newblock \href {https://doi.org/10.1007/978-3-662-43948-7_79} {\path{doi:10.1007/978-3-662-43948-7_79}}.

\bibitem[BTZ10]{BergelsonTaoZiegler2010}
Vitaly Bergelson, Terence Tao, and Tamar Ziegler.
\newblock An inverse theorem for the uniformity seminorms associated with the action of \( \mathbb{F}^\omega \).
\newblock {\em Geometric and Functional Analysis}, 19(6):1539--1596, 2010.
\newblock \href {https://doi.org/10.1007/s00039-010-0051-1} {\path{doi:10.1007/s00039-010-0051-1}}.

\bibitem[BvDH25]{bao2025tolerant}
Zongbo Bao, Philippe van Dordrecht, and Jonas Helsen.
\newblock Tolerant testing of stabilizer states with a polynomial gap via a generalized uncertainty relation.
\newblock In {\em Proceedings of the 57th Annual ACM Symposium on Theory of Computing (STOC)}, Prague, Czech Republic, 2025.
\newblock To appear. Available at \url{https://arxiv.org/abs/2403.12706}.

\bibitem[CGYZ25]{chen2025stabilizer}
Sitan Chen, Weiyuan Gong, Qi~Ye, and Zhihan Zhang.
\newblock Stabilizer bootstrapping: A recipe for efficient agnostic tomography and magic estimation.
\newblock In {\em Proceedings of the 57th Annual ACM Symposium on Theory of Computing (STOC)}, Prague, Czech Republic, 2025.
\newblock To appear. Available at \url{https://arxiv.org/abs/2408.06967}.

\bibitem[ET12]{eisner2012large}
Tanja Eisner and Terence Tao.
\newblock Large values of the {G}owers-{H}ost-{K}ra seminorms.
\newblock {\em Journal d'Analyse Mathématique}, 117:133--186, 2012.
\newblock \href {https://doi.org/10.1007/s11854-011-0033-6} {\path{doi:10.1007/s11854-011-0033-6}}.

\bibitem[GGMT25]{GowersGMT2025}
W.~T. Gowers, Ben Green, Freddie Manners, and Terence Tao.
\newblock On a conjecture of {M}arton.
\newblock {\em Annals of Mathematics}, 201(2):515--549, 2025.
\newblock Available at \url{https://annals.math.princeton.edu/2025/201-2/p04}.
\newblock \href {https://doi.org/10.4007/annals.2025.201.2.4} {\path{doi:10.4007/annals.2025.201.2.4}}.

\bibitem[GIKL23a]{grewal2023efficient}
Sabee Grewal, Vishnu Iyer, William Kretschmer, and Daniel Liang.
\newblock Efficient learning of quantum states prepared with few non-clifford gates.
\newblock {\em arXiv preprint arXiv:2305.13409}, 2023.
\newblock URL: \url{https://arxiv.org/abs/2305.13409}.

\bibitem[GIKL23b]{Grewal2023}
Sabee Grewal, Vishnu Iyer, William Kretschmer, and Daniel Liang.
\newblock Improved stabilizer estimation via bell difference sampling.
\newblock In {\em Proceedings of the 55th Annual ACM Symposium on Theory of Computing (STOC)}, pages 506--515, 2023.
\newblock \href {https://doi.org/10.1145/3618260.3649738} {\path{doi:10.1145/3618260.3649738}}.

\bibitem[GL89]{GoldreichLevin1989}
Oded Goldreich and Leonid~A. Levin.
\newblock A hard-core predicate for all one-way functions.
\newblock In {\em Proceedings of the 21st Annual ACM Symposium on Theory of Computing (STOC)}, pages 25--32, 1989.
\newblock URL: \url{https://dl.acm.org/doi/10.1145/73007.73010}, \href {https://doi.org/10.1145/73007.73010} {\path{doi:10.1145/73007.73010}}.

\bibitem[GMC14]{Garcia2014Geometry}
Héctor~J. García, Igor~L. Markov, and Andrew~W. Cross.
\newblock On the geometry of stabilizer states.
\newblock {\em Quantum Information and Computation}, 14(7-8):683--720, 2014.
\newblock URL: \url{https://arxiv.org/abs/1711.07848}, \href {https://doi.org/10.26421/QIC14.7-8-9} {\path{doi:10.26421/QIC14.7-8-9}}.

\bibitem[GNW21]{Gross2021}
David Gross, Sepehr Nezami, and Michael Walter.
\newblock Schur--{W}eyl duality for the clifford group with applications: Property testing, a robust {H}udson theorem, and de {F}inetti representations.
\newblock {\em Communications in Mathematical Physics}, 385(3):1325--1393, 2021.
\newblock \href {https://doi.org/10.1007/s00220-021-04118-7} {\path{doi:10.1007/s00220-021-04118-7}}.

\bibitem[Gow98]{Gowers1998}
W.~T. Gowers.
\newblock A new proof of {S}zemer\'edi's theorem for arithmetic progressions of length four.
\newblock {\em Geometric and Functional Analysis}, 8(3):529--551, 1998.
\newblock \href {https://doi.org/10.1007/s000390050065} {\path{doi:10.1007/s000390050065}}.

\bibitem[Gow01]{Gowers2001}
W.~T. Gowers.
\newblock A new proof of {S}zemer\'edi's theorem.
\newblock {\em Geometric and Functional Analysis}, 11(3):465--588, 2001.
\newblock \href {https://doi.org/10.1007/s00039-001-0332-9} {\path{doi:10.1007/s00039-001-0332-9}}.

\bibitem[GT08]{GreenTao2008}
Ben Green and Terence Tao.
\newblock An inverse theorem for the gowers {$U^3$}-norm.
\newblock {\em Proceedings of the Edinburgh Mathematical Society}, 51(1):73--153, 2008.
\newblock URL: \url{https://www.cambridge.org/core/journals/proceedings-of-the-edinburgh-mathematical-society/article/abs/an-inverse-theorem-for-the-gowers-u3g-norm/0A8F67E92DC546D9F27F4B71797F974C4}, \href {https://doi.org/10.1017/S0013091505000325} {\path{doi:10.1017/S0013091505000325}}.

\bibitem[HHL19]{HatamiHatamiLovett2019}
Hamed Hatami, Pooya Hatami, and Shachar Lovett.
\newblock {\em Higher-Order Fourier Analysis and Applications}, volume~13 of {\em Foundations and Trends in Theoretical Computer Science}.
\newblock Now Publishers, 2019.
\newblock \href {https://doi.org/10.1561/0400000064} {\path{doi:10.1561/0400000064}}.

\bibitem[Hoe63]{Hoeffding:1963}
Wassily Hoeffding.
\newblock Probability inequalities for sums of bounded random variables.
\newblock {\em Journal of the American Statistical Association}, 58(301):13--30, 1963.
\newblock URL: \url{https://www.tandfonline.com/doi/abs/10.1080/01621459.1963.10500830}, \href {https://doi.org/10.1080/01621459.1963.10500830} {\path{doi:10.1080/01621459.1963.10500830}}.

\bibitem[KLT23]{kim2023cubic}
Dain Kim, Anqi Li, and Jonathan Tidor.
\newblock Cubic {G}oldreich-{L}evin.
\newblock In {\em Proceedings of the 2023 ACM-SIAM Symposium on Discrete Algorithms (SODA)}, pages 4848--4866. SIAM, 2023.
\newblock \href {https://doi.org/10.1137/1.9781611977554.ch178} {\path{doi:10.1137/1.9781611977554.ch178}}.

\bibitem[Lov17]{Lovett2017}
Shachar Lovett.
\newblock Additive combinatorics and its applications in theoretical computer science.
\newblock {\em Theory of Computing}, 13(4):1--55, 2017.
\newblock URL: \url{https://theoryofcomputing.org/articles/gs008/}, \href {https://doi.org/10.4086/toc.gs.2017.008} {\path{doi:10.4086/toc.gs.2017.008}}.

\bibitem[O'D14]{ODonnell2014}
Ryan O'Donnell.
\newblock {\em Analysis of Boolean Functions}.
\newblock Cambridge University Press, 2014.
\newblock URL: \url{https://www.cambridge.org/core/books/analysis-of-boolean-functions/B05A66E4DCC778E02B84C16376F4D1FD}, \href {https://doi.org/10.1017/CBO9781139814782} {\path{doi:10.1017/CBO9781139814782}}.

\bibitem[Rot53]{Roth1953}
Klaus~F. Roth.
\newblock On certain sets of integers.
\newblock {\em Journal of the London Mathematical Society}, 28(1):104--109, 1953.
\newblock \href {https://doi.org/10.1112/jlms/s1-28.1.104} {\path{doi:10.1112/jlms/s1-28.1.104}}.

\bibitem[Sam07]{Samorodnitsky2007}
Alex Samorodnitsky.
\newblock Low-degree tests at large distances.
\newblock In {\em Proceedings of the 39th Annual ACM Symposium on Theory of Computing (STOC)}, pages 506--515, 2007.
\newblock URL: \url{https://dl.acm.org/doi/10.1145/1250790.1250864}, \href {https://doi.org/10.1145/1250790.1250864} {\path{doi:10.1145/1250790.1250864}}.

\bibitem[San12]{Sanders2012}
Tom Sanders.
\newblock On the {B}ogolyubov-{R}uzsa lemma.
\newblock {\em Analysis \& PDE}, 5(3):627--655, 2012.
\newblock \href {https://doi.org/10.2140/apde.2012.5.627} {\path{doi:10.2140/apde.2012.5.627}}.

\bibitem[Sze75]{Szemeredi1975}
Endre Szemerédi.
\newblock On sets of integers containing no $k$ elements in arithmetic progression.
\newblock {\em Acta Arithmetica}, 27:199--245, 1975.
\newblock \href {https://doi.org/10.4064/aa-27-1-199-245} {\path{doi:10.4064/aa-27-1-199-245}}.

\bibitem[Tao12]{Tao2012}
Terence Tao.
\newblock {\em Higher-Order Fourier Analysis}, volume 142 of {\em Graduate Studies in Mathematics}.
\newblock American Mathematical Society, 2012.
\newblock URL: \url{https://bookstore.ams.org/gsm-142}, \href {https://doi.org/10.1090/gsm/142} {\path{doi:10.1090/gsm/142}}.

\bibitem[Tre09]{Trevisan2009}
Luca Trevisan.
\newblock Additive combinatorics and theoretical computer science.
\newblock {\em SIGACT News}, 40(2):50--72, 2009.
\newblock URL: \url{https://theory.stanford.edu/~trevisan/pubs/addcomb-sigact.pdf}, \href {https://doi.org/10.1145/1556154.1556170} {\path{doi:10.1145/1556154.1556170}}.

\bibitem[TV06]{TaoVu2006}
Terence Tao and Van~H. Vu.
\newblock {\em Additive Combinatorics}, volume 105 of {\em Cambridge Studies in Advanced Mathematics}.
\newblock Cambridge University Press, 2006.
\newblock \href {https://doi.org/10.1017/CBO9780511755149} {\path{doi:10.1017/CBO9780511755149}}.

\bibitem[TW14]{TulsianiW2014}
Madhur Tulsiani and Julia Wolf.
\newblock Quadratic {G}oldreich--{L}evin theorems.
\newblock {\em SIAM Journal on Computing}, 43(2):730--766, 2014.
\newblock \href {https://doi.org/10.1137/12086827X} {\path{doi:10.1137/12086827X}}.

\bibitem[TZ12]{TaoZiegler2012}
Terence Tao and Tamar Ziegler.
\newblock The inverse conjecture for the gowers norm over finite fields in low characteristic.
\newblock {\em Annals of Combinatorics}, 16(1):121--188, 2012.
\newblock \href {https://doi.org/10.1007/s00026-012-0129-0} {\path{doi:10.1007/s00026-012-0129-0}}.

\bibitem[vD25]{vanDordrecht2025}
P.J. van Dordrecht.
\newblock Gowers {$U^3$} inverse theorem for quantum states.
\newblock Msc thesis, University of Amsterdam, Amsterdam, The Netherlands, 2025.
\newblock Supervisor: J. Briët.
\newblock URL: \url{https://scripties.uba.uva.nl/search?id=record_55717}.

\bibitem[Zha23]{Zhao2023}
Yufei Zhao.
\newblock {\em Graph Theory and Additive Combinatorics: Exploring Structure and Randomness}.
\newblock Cambridge University Press, 2023.
\newblock \href {https://doi.org/10.1017/9781009310956} {\path{doi:10.1017/9781009310956}}.

\end{thebibliography}
\end{document}